\documentclass[lettersize,journal]{IEEEtran}
\usepackage{amsmath,amsfonts}
\usepackage{array}
\usepackage{textcomp}
\usepackage{stfloats}
\usepackage{url}
\usepackage{verbatim}
\usepackage{cite}
\usepackage{times}
\usepackage{amsthm}
\usepackage{amsmath}
\usepackage{psfrag}
\usepackage{epsfig}
\usepackage{graphicx}
\usepackage{color}
\usepackage{bm}
\usepackage{epstopdf}
\usepackage{verbatim}
\usepackage{graphicx}   

\usepackage{textcomp}

\usepackage{latexsym,verbatim}
\usepackage{psfrag,epsfig}
\usepackage{graphicx,subfigure}
\usepackage{listings}
\usepackage{tikz}
\usepackage{url}
\usepackage{algorithmic}

\usepackage{booktabs}

\newtheorem{theorem}{Theorem}
\newtheorem{lemma}{Lemma}
\newtheorem{prop}{Proposition}
\newtheorem{definition}{Definition}

\usepackage[linesnumbered,ruled]{algorithm2e}
\SetAlFnt{\footnotesize} 
\SetKwBlock{KwFunction}{function}{}
\SetKwBlock{KwOn}{on}{}
\SetKwBlock{KwIf}{if}{}
\SetKwBlock{KwContinue}{continue}{}
\SetKwBlock{KwPeriodically}{periodically}{}
\SetKwBlock{KwState}{state}{}
\SetKwBlock{KwAtomic}{atomically}{}
\SetKwBlock{KwConcurrently}{concurrently}{}
\SetCommentSty{textup} 
\SetKwComment{KwIttayComment}{(}{)} 
\SetKw{KwAnd}{and}
\SetKw{KwOr}{or}
\SetKw{KwSetTimer}{setTimer}
\SetKwBlock{KwExpTimer}{on timer expiration}{}
\SetKwFor{ForAll}{forall}{}

\newcommand{\ADD}[1]{\textcolor{black}{#1}}

\newenvironment{proof-sketch}{\noindent{\bf Sketch of Proof:}\hspace*{1em}}{\qed\bigskip}

\begin{document}
	
	\title{TIPS: Transaction Inclusion Protocol with Signaling  in DAG-based Blockchain}
	
	\author{
		Canhui Chen,
		Xu Chen,
		Zhixuan Fang
		\IEEEcompsocitemizethanks{

			\IEEEcompsocthanksitem Canhui Chen is with  Institute for Interdisciplinary Information Sciences (IIIS), Tsinghua University, Beijing, China. 
			E-mail: chen-ch21@mails.tsinghua.edu.cn.
			\IEEEcompsocthanksitem Xu Chen is with School of Computer Science and Engineering, Sun Yat-sen University, Guangzhou, China
			E-mail: chenxu35@mail.sysu.edu.cn.
			\IEEEcompsocthanksitem  Zhixuan Fang is with  Institute for Interdisciplinary Information Sciences (IIIS), Tsinghua University, Beijing, China, and  Shanghai Qi Zhi Institute, Shanghai, China. E-mail: zfang@mail.tsinghua.edu.cn.
		    \IEEEcompsocthanksitem 
			Corresponding author: Zhixuan Fang.
		}
	}
	
	
	
	
	\maketitle
	
	\begin{abstract}
		Directed Acyclic Graph (DAG) is a popular approach to achieve scalability of blockchain networks. Due to its high efficiency in data communication and great scalability, DAG has been widely adopted in many applications such as Internet of Things (IoT) and Decentralized Finance (DeFi). DAG-based blockchain, nevertheless, faces the key challenge of transaction inclusion collision due to the high concurrency and the network delay. Particularly, the transaction inclusion collision in DAG-based blockchain leads to the revenue and throughput dilemmas, which would greatly degrade the system performance. In this paper, we propose ``TIPS'', the Transaction Inclusion Protocol with Signaling, which broadcasts a signal indicating the transactions in the block. We show that with the prompt broadcast of a signal, TIPS substantially reduces the transaction collision and thus resolves these dilemmas. Moreover, we show that TIPS can defend against both the denial-of-service and the delay-of-service attacks. We also conduct intensive experiments to demonstrate the superior  performance of the proposed protocol.
	\end{abstract}
	
	\begin{IEEEkeywords}
		blockchain, game theory, performance analysis
	\end{IEEEkeywords}

	\section{Introduction}\label{sec:intro}
	
	\IEEEPARstart{B}{lockchain} plays an important role in many fields such as finance, supply chain, and IoT services, thanks to its advantages of decentralization, persistence, pseudonymity, and auditability \cite{swan2019blockchain,xiong2018cloud,kang2019toward}. However, the poor performance of (mostly) linear blockchains (e.g., Bitcoin, Ethereum) severely limits the possible large-scale applications in practice  \cite{zhou2020solutions}. 
	To address this challenge, Directed Acyclic Graph (DAG) \cite{lewenberg2015inclusive} is introduced as an alternate structure to address the issue of the scalability of blockchain. DAG-based blockchain allows multiple blocks or transactions to be appended concurrently without solving the forks, which achieves high efficiency in data communication and great scalability. Due to these properties, DAG-based blockchain has been widely adopted in many IoT scenarios \cite{wang2021understanding,cao2019internet,zhang2021dag,yang2020ldv,li2020direct}. 
	%
	DAG-based blockchain, nevertheless, faces the key challenge of transaction inclusion collision \cite{lewenberg2015inclusive}. 
	Due to the high concurrency and the network  delay, miners usually do not have complete information of the updated blockchain. Thus, miners often include the same transactions in concurrent blocks, generating redundant records in the blockchain. 
	The collision in transaction inclusion wastes the block capacity and severely degrades the system performance \cite{gupta2019cdag}. 

	Though DAG-based blockchain systems admit high concurrency in transaction process, the risk of transaction collision indeed induces dilemmas on miners' revenue and system throughput.
	The \emph{revenue dilemma} indicates the difficult situation in improving  miners' revenue, or, in achieving higher transaction fees per block.
	When miners are packing transactions into a block, the intuitive strategy is to select transactions with high rewards (fees).
	But concurrent blocks with the same set of high-fee transactions lead to the split of reward among miners, which discourages miners from including these transactions.
	Indeed, the theoretical analysis in \cite{lewenberg2015inclusive} has verified that miners would avoid transaction collision by choosing transactions with lower fees.
	Such a negative impact would also degenerate the profit efficiency and the quality of service of the system, as the high value transactions will not be processed with priority.
	
	On the system level, the \emph{throughput dilemma} also arises and implies the difficulty in scaling up the system. Since the network propagation delay is positively correlated with the block size in the current DAG-based blockchain, packing more transactions in a single block leads to a larger network propagation delay. 
	The increased network propagation time induces more collisions and thus wastes the system throughput. 
	In fact, a similar dilemma in increasing block size has been observed  in linear blockchain systems where network propagation delay leads to forking \cite{gobel2017increased}.
	
	To tackle these issues, in this paper, we propose the Transaction Inclusion Protocol with Signaling (TIPS), where we introduce a lightweight ``signal'' mechanism in the DAG-based blockchain.
	In TIPS, when a miner mines a new block, he will first broadcast a ``signal'' before broadcasting the block. The other miners will adjust their transaction inclusion strategies reacting to the ``signal''. 
	A signal has the following two properties:
	(i) The size of the signal is small enough for quick propagation in the network.
	(ii) The signal provides information about transactions included in the corresponding block. 
	
	
	Specifically, we use a Bloom filter \cite{bloom1970space}  to generate the ``signal'' in the block header. The Bloom filter can indicate whether a transaction is in the block. When a miner mines a new block, it should first broadcast the block header as a signal to the whole network. Since the size of the signal is small, other miners can receive the signal in a short time. 
	Based on the signal, other miners can infer which transactions are included in the new block with high probability. 
	Thus, other miners are able to avoid the transaction inclusion collision.
	
	The introduction of TIPS can break down the revenue dilemma and throughput dilemma. We show this by analyzing the miners' transaction inclusion strategy under TIPS. We model the transaction inclusion process as a non-cooperative game and investigate the equilibrium strategy.
	We show that with the signal information,  miners are able  to select high value transactions in equilibrium, which dissolves the revenue dilemma. 
	In the meantime, the signal in TIPS effectively reduces the network propagation time, which allows the implementation of a huge block size, and thus, breaks the throughput dilemma.
	These results are supported by both theoretical analysis and experimental validation.
	Beyond the discussion on performance,  we also investigate the security of TIPS and show that TIPS can defend against existing classic attacks.

	The key contributions of the paper are listed as follows:
	
	\begin{itemize}
		\item We characterize the revenue dilemma and the throughput dilemma of DAG-based blockchain systems. Specifically, miners are suffering from low revenue no matter the transaction inclusion strategy, while the system throughput does not increase in the block size. We show this is due to the transaction collisions in the DAG-based blockchain system, and that a low network propagation delay is the key to breaking the dilemmas.
		\item We propose a novel Transaction Inclusion Protocol with Signaling (TIPS) in the DAG-based blockchain. 
		TIPS includes a Bloom filter in the block header, which serves as a signal in the mining process and can indicate the transaction included in the newly-mined block. By separating the propagation/verification processes of the block header and block body, TIPS can broadcast the block header (signal) as soon as possible, which significantly lowers the effective network propagation delay.
		\item We provide a thorough theoretical analysis of the performance and security of TIPS. 
		We adopt a game-theoretic framework to show that at the Nash equilibrium, with the prompt signal from TIPS, the system effectively lowers the network propagation delay, and thus reduces transaction collisions.	
		Besides, we also develop a DAG-based blockchain simulator and conduct intensive experiments. Both the theoretical analysis and experiment results show that TIPS can substantially resolve both the revenue and the throughput dilemmas.
	\end{itemize}
	
	The rest of the paper is organized as follows. 
	In Section \ref{sec:dilemmas}, we introduce the miners' transaction inclusion game and characterize the two dilemmas, namely, the revenue dilemma and the throughput dilemma.
	In Section \ref{sec:system_model}, we introduce TIPS, the transaction inclusion protocol with signaling in the DAG-based blockchain. 	
	In Section \ref{sec:tips_break}, we show that TIPS can break down the aforementioned dilemmas.
	In Section \ref{sec:experiment}, we conduct intensive experiments to demonstrate the efficiency of TIPS. In Section \ref{sec:security}, we discuss the possible security threats in the TIPS. In Section \ref{sec:related}, we review related literature. 	Section \ref{sec:conclusion} concludes the paper with final remark.

	\section{Transaction Inclusion Game and the Dilemmas}\label{sec:dilemmas}
	
	In this section, we first investigate the system model and the transaction inclusion game. We next show that there are two dilemmas in the DAG-based blockchain, namely, the revenue dilemma and the throughput dilemma. From the miners' perspective, the revenue dilemma indicates the difficult situation in improving the miners' revenue. From the system perspective, the throughput dilemma reveals the difficulty in scaling up the system.
	\ADD{We list the key notations in Section \ref{sec:dilemmas} and Section \ref{sec:tips_break} in Table \ref{table:notation_table}.}
	
		\begin{table}[!t]
		\caption{Summary of Notations}	\label{table:notation_table}
		\begin{tabular}{@{}ll@{}}
			\toprule
			Notation                & Description          \\ \midrule
			$b$                  & Number of bits in the Bloom filter  \\
			$h$                    & Number of hash functions in the Bloom filter                  \\
			$\epsilon$                    & The probability of false positives of the Bloom filter \\
			$\eta$                    & The probability of rejecting a valid Bloom filter \\
			$X$               &  Number of bits that are set to 1 in the Bloom filter \\
			$n$                    & The maximum number of transactions in a block\\
			$m$              & Size of the transaction pool\\
			$\lambda$              & Block production rate            \\
			$\Delta$  & Propagation delay time \\ 
			$f_i$  & Transaction fee of the $i$-th transaction\\
			$\textbf{p}$ & Transaction inclusion strategy \\
			$p_i$ & Probability of including the transaction $i$ in the strategy \\
			\midrule                                    
		\end{tabular}
	\end{table}
	
	\subsection{Transaction Inclusion Game}\label{sec:game}
	
	In this paper, we consider a DAG-based blockchain, where the block generation process follows the Poisson process with a rate $\lambda$. 
	\ADD{Following the common assumption in previous literature (e.g., \cite{hari2019accel, lewenberg2015inclusive}),  we denote the maximum network propagation delay for a block as $\Delta$, which means all nodes in the blockchain network can receive the newly-mined block after $\Delta$.
		Moreover, similar to \cite{lewenberg2015inclusive,chen2020nonlinear}, we assume that miners in the system are homogeneous in mining power and share the identical transaction pool. 
		The transaction pool includes at most $m$ transactions. Due to the block size limit, each block can contain at most $n$ transactions.\footnote{Without loss of generality, we can always assume that there are $m$ transactions in the pool and each block will contain $n$ transactions. Practically, if the transactions are insufficient, we can assume that there are dummy transactions with zero transaction fees, which does not affect our analysis.}}
	To maximize the mining revenue, the miner will decide his transaction inclusion strategy and pick up some transactions from the transaction pool and pack them into the mining block to earn the transaction fees. 
	
	We denote a miner's transaction inclusion strategy as $\textbf{p} \in \mathbb{P}$.
	Here $\mathbb{P} =  \{ \textbf{p} | 0 \leq p_i \leq 1 \text{ and } \sum_{i=1}^{m} p_i = n \}$ denotes the set of transaction inclusion strategies, and $p_i$ denotes the probability of including the transaction $i$ in the new block.
	Without loss of generality, we sort the transactions in the transaction pool in descending order by their transaction fees, where the transaction fee of the transaction $i$ is denoted as $f_i$. Then we have $f_1 \geq f_2 \geq \dots \geq f_m$.
	As examples, we show below three typical transaction inclusion strategies:
	\begin{itemize}
		\item Random inclusion ($\textbf{p}^{\rm rand}$): $p_1 = p_2 = \dots = p_m = \frac{n}{m}$.
		\item Random inclusion with priority ($\textbf{p}^{\rm priority}$): $p_1 \geq p_2 \geq \dots \geq p_m$ and $\frac{p_1}{f_1} = \frac{p_2}{f_2} = \dots = \frac{p_m}{f_m}$.
		\item Top $n$ ($\textbf{p}^{\rm top}$): $p_1 = p_2 = \dots = p_n = 1$ and $p_{n+1} = p_{n+2} = \dots = p_m = 0$.
	\end{itemize}
	
	Following the similar setting in \cite{lewenberg2015inclusive}, we model the miners' transaction selection as a single-shot  transaction inclusion game \cite{chen2020nonlinear}, where each miner adopts his transaction inclusion strategy to maximize his revenue. 
	Since the coinbase transaction reward is independent of the miner's transaction selection, for simplicity, we only consider the transaction fee reward in the miner’s revenue.
	
	The miner's revenue is analyzed as follows. Consider a miner $A$ that finds and propagates a block, which includes a transaction $i$.  If $\iota$ more miners have successfully mined a block with the same transaction $i$ during the network propagation period of miner $A$'s block, we assume that this miner's expected reward from this transaction $i$ is $f_i/ (\iota+1)$, i.e., assuming an equal network advantage for all miners.
	Note that such a model of probabilistic and homogeneous network advantage during the propagation period is common in the previous study (e.g., \cite{lewenberg2015inclusive,eyal2014majority}).
	Then the miner's expected reward  is analyzed in the following lemma.
	\begin{lemma}\label{lemma:reward}
		The miner's revenue on one block with transaction inclusion strategy $\textbf{p}$ given that all the other miners adopt the strategy $\textbf{p}'$ is
		\begin{equation}\nonumber
			R(\textbf{p}|\textbf{p}') = \sum_{i=1}^{m} p_i f_i r(p_i'|\Delta),
		\end{equation}
		where 
		\begin{equation}\nonumber
			r(p_i' | \Delta) = \frac{\left(1 - e^{-\lambda \Delta p_i'}\right) }{\lambda \Delta p_i'}.
		\end{equation}
	\end{lemma}
	The proof of Lemma \ref{lemma:reward} is shown in Appendix \ref{proof:lemma:reward}. Besides, we postpone all proofs of this paper into the Appendix. 
	The term $f_i \cdot r(p_i' | \Delta)$  in Lemma \ref{lemma:reward} reflects the expected reward for a miner to include transaction $i$ when all other miners include transaction $i$ in their blocks with probability $p_i'$, given network propagation delay being $\Delta$.
	%
	Specially, when all miners adopt the same transaction inclusion strategy $\textbf{p}$ (i.e., the symmetric case) and there is no futher confusion, we denote the miner's revenue as
	$R(\textbf{p}) = R(\textbf{p}|\textbf{p})$ for simplicity.
	
	To study the stable transaction inclusion behavior of miners, we first define the  Nash equilibrium and approximate Nash equilibrium of the transaction inclusion game below.
	
	\begin{definition}\label{def:NE}
		The transaction inclusion strategy $\textbf{p}^*$ is a $\mathbf{\xi}$-\textbf{approximate  Nash equilibrium} of the transaction inclusion game if 
		\begin{equation}\nonumber
			R(\textbf{p}^*) \geq \max_{\textbf{p} \in \mathbb{P}} R(\textbf{p} | \textbf{p}^*) - \xi .
		\end{equation}
		Specially, when $\xi = 0$, i.e., $R(\textbf{p}^*) \geq \max_{\textbf{p} \in \mathbb{P}} R(\textbf{p} | \textbf{p}^*)$, the transaction inclusion strategy $\textbf{p}^*$ is a \textbf{Nash equilibrium}.
	\end{definition}
	
	The symmetric equilibrium in the single-shot transaction inclusion game has been analyzed in \cite{lewenberg2015inclusive}, which is shown in the following theorem.
	
	\begin{theorem}\label{th:NE_general}
		(\cite{lewenberg2015inclusive}) With the network propagation delay for the whole block as $\Delta$, the symmetric equilibrium strategy of the transaction inclusion game is $p^*(\Delta)$, where
		\begin{equation}\nonumber
			\begin{aligned}
				p_i^*(\Delta) = \begin{cases}
					\min \{r^{-1}\left( \frac{c_{l^*} }{f_i} | \Delta \right), 1  \}, & 1 \leq i \leq l^*,\\
					0, & l^* < i \leq m.
				\end{cases} \\
			\end{aligned}
		\end{equation}
		Also, we have
		\begin{equation}\nonumber
			\begin{aligned}
				&F_l(c) = \sum_{i=1}^{l} \min \{ r^{-1}(c/f_i | \Delta),1 \} - n, \quad \forall 1 \leq l \leq m, \\
				&l^* = \max \{ l \leq m | \forall i \leq l: F_i\left( f_i \right) \leq 0  \}, \\
				&c_{l^*} \text{ is the root of } F_{l^*}.
			\end{aligned}
		\end{equation}
	\end{theorem}

	\subsection{Revenue Dilemma Analysis}\label{sec:profit}
	As discussed in Section \ref{sec:intro}, to achieve a high revenue, miners are supposed to include transactions with high fees, i.e.,  to adopt the ``top $n$'' strategy. However, as DAG-based blockchain allows high concurrency of block generation, every miner choosing high-fee transactions  will result in severe transaction inclusion collision, which will degrade miners' revenue. This is the revenue dilemma due to the collision of included transactions.

	The equilibrium strategy in Theorem \ref{th:NE_general} is highly related to the network propagation delay and is indeed a compromised solution facing the revenue dilemma. We can find that the equilibrium transaction inclusion strategy can be considered as the ``soft top $n$ strategy'', since only the transaction with top $l^*$ transaction fees might be included in a new block.
	A large network propagation delay $\Delta$  in Theorem \ref{th:NE_general} will result in a large $l^*$. 
	Besides, when the $l^*$ is large, for example, $l^* = m$, the ``soft top $n$ strategy'' performs like the random strategy $\textbf{p}_{\rm rand}$. 
	Below, we show that the random strategy is a $\xi$-approximate Nash equilibrium strategy, and $\xi$ represents the gap between the random strategy and the equilibrium strategy.
	\begin{theorem}\label{th:approx_rand}
		The random strategy, i.e., $\textbf{p}^{\rm rand}$ is a $\xi$-approximate Nash equilibrium, where
		\begin{equation}\nonumber
			\xi = n \frac{1 - e^{-\lambda \Delta \frac{n}{m}}}{\lambda \Delta \frac{n}{m}} \left( \frac{1}{n} \sum_{i=1}^{n} f_i - \frac{1}{m} \sum_{i=1}^{m} f_i   \right).
		\end{equation}
		Specially, when $\Delta \rightarrow \infty$, the random strategy is the Nash equilibrium.
	\end{theorem}
	
	\ADD{For convenience, we denote $y(\Delta) = \frac{1 - e^{-\lambda \Delta \frac{n}{m}}}{\lambda \Delta \frac{n}{m}}$. Note that $y(\Delta)$ is monotonically decreasing in $\Delta$, which implies that the extra revenue from deviating from the random strategy will decrease as the network propagation delay becomes larger. Thus the equilibrium transaction inclusion strategy will lean towards the random strategy. }
	Since broadcasting the whole block in the blockchain network takes a long time, the equilibrium strategy in the practical scenario is similar to the random strategy.
	However, the random strategy ignores the difference between different levels of the transaction fees, resulting in the dilemma where miners can not include transactions with high transaction fees. This is the revenue dilemma in DAG-based blockchain.

	\subsection{Throughput Dilemma Analysis}\label{sec:throughput}
	
	In this section, we study the throughput dilemma. The system throughput is measured by transaction per second (TPS). Intuitively, to achieve high throughput, the block size should be as large as possible so as to include more transactions in a block. However, a large block size will lead to a large network propagation delay and result in lots of transaction inclusion collisions, which will limit the system throughput. This is the throughput dilemma in DAG-based blockchain.

	Here, we first analyze the block capacity utilization and system throughput, and further characterize the throughput dilemma in DAG-based blockchain
	We define the block capacity utilization as the ratio of the number of unique transactions included in the blocks and the number of total transactions included in the blocks in long term. The system throughput is measured by transaction per second (TPS).
	Then the system throughput of the DAG-based blockchain is analyzed in the following theorem.
	
	\begin{theorem}\label{th_tps}
		The block capacity utilization and the throughput of the DAG-based blockchain with the transaction inclusion strategy $\textbf{p}$ and the network propagation delay $\Delta$ are
		\begin{equation}\nonumber
			\begin{aligned}
				&\text{U}(\textbf{p}) = \frac{   m - \sum_{i=1}^{m} (1 - p_i) e^{-\lambda \Delta p_i}  }{n(\lambda \Delta + 1)}, \quad \text{TPS}(\textbf{p}) = \lambda n \text{U}(\textbf{p}),
			\end{aligned}
		\end{equation}
		respectively.
	\end{theorem}

	From Theorem \ref{th_tps}, we can find that to achieve a higher system throughput, we should enlarge the block size so as to include more transactions. However, if the block size is increased, the corresponding network propagation delay $\Delta$ will also increase, leading to a lower block capacity utilization, which will further degrade the system throughput. On the other hand, though decreasing the block size can reduce the network propagation delay and thus increase the block capacity utilization, there is only a small amount of transactions can be included in one block, which limits the system throughput. This is the throughput dilemma in DAG-based blockchain.

	\section{Transaction Inclusion Protocol with Signaling}\label{sec:system_model}
	
	
	To tackle the dilemmas in DAG-based blockchain, we introduce ``TIPS'', i.e., the Transaction Inclusion Protocol with Signaling. 
	The key features of TIPS are (1) TIPS introduces a signal to indicate the transactions included in the block. (2) TIPS broadcast the signal earlier than the whole block.
	
	As a baseline, we will compare TIPS with the standard DAG-based blockchain protocol (i.e., without TIPS), where the miners do not obtain any information of the newly-generated block until they receive the whole block.

	\subsection{Bloom Filter in Block Header}\label{sec:bloom}
	
	A Bloom filter \cite{bloom1970space} is a probabilistic data structure that answers the query of ``whether an element is a member of a set''.
	It returns a binary answer, i.e., either ``True'' or ``False''. 
	If the Bloom filter returns ``False'', the element is definitely not in the set. 
	However, if the Bloom filter returns ``True'', the element is a member of the set with high probability.

	As the most important feature, TIPS introduces a Bloom filter in the header of each block to maintain the information of transactions in this block. 
	Thus, the Bloom filter can answer the query of whether a transaction is in the corresponding block quickly and with high probability.

	Next, we examine the performance of a Bloom filter. The key metric is the false positive rate of a Bloom filter, i.e., the probability that the Bloom filter returns ``True'' but the element is not a member of the set. 
	Consider a Bloom filter with $b$ bits and $h$ different hash functions. 
	We assume that the block size limit is $n$ transactions per block. Thus, we can insert at most $n$ transactions into the Bloom filter associated with the block. The probability of false positives of the Bloom filter with $n$ transactions is 
	\begin{equation}\label{eq:bloom_false}
		\epsilon = \left( 1 - e^{-hn / b} \right)^h.
	\end{equation}
	
	It is worth mentioning that the representation of one transaction in the Bloom filter always takes a fixed space (i.e., $b$ bits) regardless of the size of the transactions.
	The size of a Bloom filter is drastically smaller than the size of the original block body, due to the property of hash functions. 
	In the current Bitcoin network, the average block size is 1.2MB, the average transaction size is 500B, and the average number of transactions in a block is 2500. As an example, for a Bloom filter with $b=20000$ bits and $h=5$ hash functions, the false positive probability is $2.17\%$, which is low enough for practice. Meanwhile, the total size of the signal is only 2.5KB, which is less than $1\%$ of the original block size and is small enough to propagate through the whole network in a short time.

	\subsection{Header-First Block Propagation}\label{sec:system_mode:block_propagation}
	
	To broadcast the signal to the whole network as soon as possible, we improve the block propagation model. \ADD{Figure \ref{fig:propagation_signal_cmp} shows the block propagation model in TIPS and the standard protocol. We can find that in the standard protocol, the node will start the propagation of a block (including the block header and the block body) only after it receives and verifies the whole block. On the contrary, TIPS decouples the propagation/verification processes of the block header and the block body. 
	In TIPS, the node will broadcast the received header once the verification is complete, without waiting to receive the whole block body. Therefore, TIPS  can broadcast the block header to the whole network in a short time.}

	\begin{figure}[!th]
		\centering
		\subfigure[Block propagation in TIPS]{
			\includegraphics[width=0.9\linewidth]{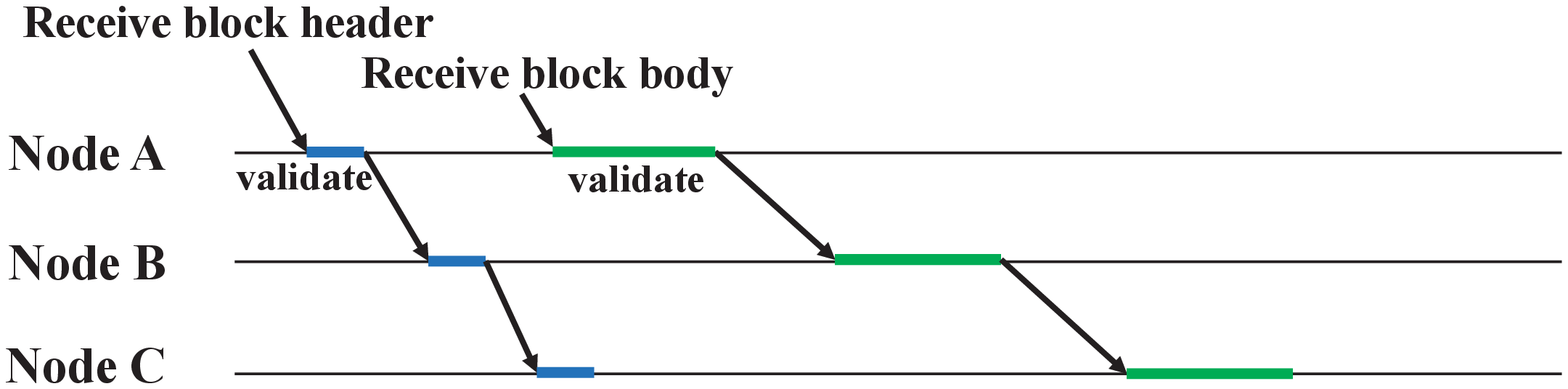}
		}
		\subfigure[Block Propagation in standard protocol]{
			\includegraphics[width=0.9\linewidth]{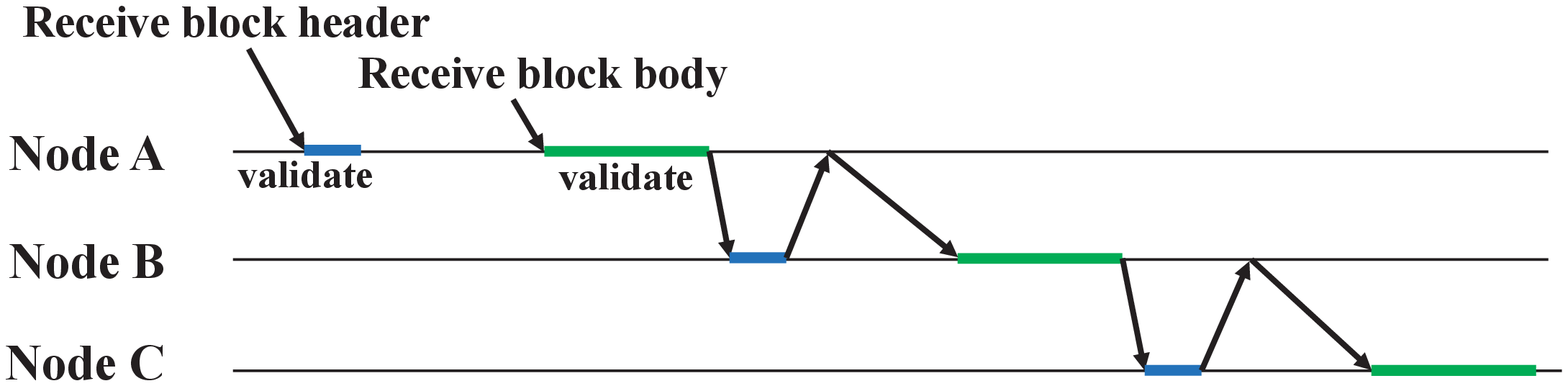}
		}
		
		\caption{\ADD{Block propagation model in TIPS and the standard protocol}}
		\label{fig:propagation_signal_cmp}
	\end{figure}
	

	{We explain the detailed block propagation model of TIPS as follows.}
	For a node $A$, upon receiving a message, there are two possible cases:
	\begin{enumerate}
		\item 	When node $A$ \textbf{receives a new block header $BH_c$} of block $B_c$, it validates $BH_c$ and checks whether the hash value of the block header satisfies the PoW puzzle.
		If the block header $BH_c$ is valid, node $A$ should broadcast $BH_c$ as soon as possible. Note that the validation process for the block header is not time-consuming and the size of the block header is small. Therefore, the block header $BH_c$ can be propagated to the whole network in a short time. 
		\item 	When node $A$ \textbf{receives a new block body $BB_c$} of block $B_c$, it validates $BB_c$ as follows:
		\begin{itemize}
			\item Block Header Existence: If the miner did not receive the corresponding block header $BH_c$ before, he should reject the block body immediately since he can not validate the PoW of the block.
			\item Bloom filter Validation: If the Bloom filter in the block header does not match the transactions in the block body, the block will be marked as ``invalid'' and be rejected.
		\end{itemize}
		If the block body is valid, node $A$ should broadcast $BB_c$ as soon as possible.
	\end{enumerate}

	\ADD{
	Besides, TIPS also helps to optimize the block transmission process from the following perspectives. 
	(i) In TIPS, the received signal can indicate the transactions in the newly-mined but unreceived block. Therefore, TIPS allows the miners to pre-verify the transactions that the signal indicates when receiving the signal. In this way, the miners can speed up the verification process when receiving the whole block, and therefore speed up the block transmission process. (ii) TIPS separates the verification/propagation processes of the block header and the block body. Since the block header is an identification of the whole block, the miners that received the block header without the block body can try to pull the block body from the connected neighbors instead of just passively waiting for the block propagation, which can also speed up the block transmission process. 
	(iii) The independent header-first propagation in TIPS also slightly accelerates the broadcast of the block body, as the block body broadcast does not need to wait for the neighbor's confirmation of the received header.
	}

	\subsection{Mining Process in TIPS}
	We denote the set of transactions hitting the Bloom filter of block $B_c$, i.e., transactions that are reported to be included in the block by the Bloom filter, 
	as $TH_c$. 
	We denote the set of transactions that are actually included in the block body of block $B_c$ as $TB_c$. Since all included transactions will be confirmed by the Bloom fitler, we have that $TB_c \subset TH_c$. Specifically, we have that $\mathbb{E}(|TB_c|)/\mathbb{E}(|TH_c|) = 1 - \epsilon$, where $\epsilon$ is the false positive probability of the Bloom filter in (\ref{eq:bloom_false}).

	\begin{figure}[!ht]
		\includegraphics[width=1.0\linewidth,height=0.55\linewidth]{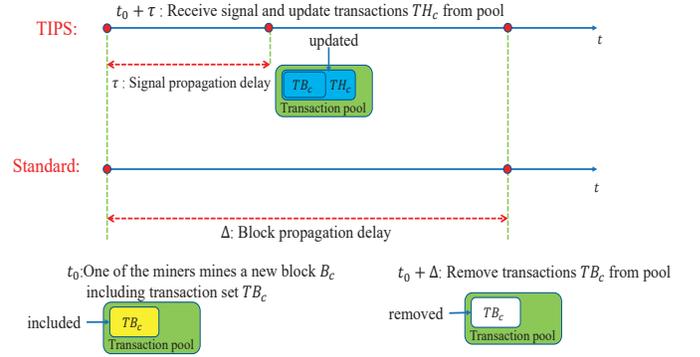}
		\caption{Mining process in TIPS}
		\label{fig:model}
	\end{figure}

	%
	In this paper, we consider the following mining reward rule, which is common in practices (e.g., \cite{eyal2014majority}). 
	If there are $k$ miners that include the same transaction $i$ with fee $f_i$ in their blocks, the transaction fee will be rewarded to the miner who first mines the block.
	Note that our analysis technique applies to more complicated rules.
	
	The mining process is shown in Algorithm \ref{alg:mining_dag}, visualized in Figure \ref{fig:model}. 
	We use $\Delta$ to denote the propagation time of a block, and $\tau$ to denote the propagation time of a block header (in which the signal is included). 
	Since the size of Bloom filter is much smaller than a block body, we have $\tau\ll \Delta$.
	$BL_c$ is the Bloom filter in $BH_c$, $TH_c$ is the set of transactions hitting $BL_c$ in transaction pool, and $TB_c$ is the set of transactions included in $BB_c$.
	

	In the following, we specify the mining process, especially the transaction selection.
	In general, a miner maintains an expected value for each transaction in his transaction pool, and selects transactions for mining according to some specific transaction inclusion strategies. 
	In TIPS, $f_i$ denotes the expected value of transaction $i$ instead of its original transaction fee. The expected value of the transaction reflects the expected reward that a miner can obtain from including this transaction in his block and  will change during the mining process. Specially, the initial expected value of the transaction is exactly the corresponding transaction fee.
	The miner updates his transaction pool and the expected value of transactions in the following two cases.
	
	
	\begin{algorithm}[!t] 
		\SetAlgoNoLine 
		\SetAlgoNoEnd 
		\DontPrintSemicolon 
		\caption{Mining process in TIPS} 
		\label{alg:mining_dag} 
		\KwOn(Receive a block header $BH_c$) { 
			\If{ the block header $BH_c$ is valid}{
				Extract Bloom filter $BL_c$ from $BH_c$ \;
				Select the transactions hitting $BL_c$ from transaction pool, i.e., $TH_c$ \;
				Update the expected value of the transaction in $TH_c$. Specifically, the expected value of transaction $tx_i$ will be $\epsilon f_i$.
			}
		} 
		\BlankLine 
		
		\KwOn(Receive a block body $BB_c$) {
			\If{the block body $BB_c$ is valid and matches the block header $BH_c$}{
				Remove the transactions that are included in the block body, i.e., $TB_c$, from the transaction pool. \;
				Update the expected value of the transactions that hit the Bloom filter in block header but are not included in the block body. Specifically, $\forall i \in (TH_c / TB_c)$, the expected value of transaction $i$ should be updated to $f_i / \epsilon $.
			}
			\Else{
				update the expected values of all the transactions that hit the Bloom filter of the block header $BH_c$ to their original transaction fee, that is, $\forall i \in TH_c$, the expected value of transaction $i$ should be updated to $f_i / \epsilon$. \;
			}
		} 
		\BlankLine 
		
		\KwOn(The block header $BB_c$ is expired){
			Mark the block header as ``expired''. \;
			Update the expected values of all the transactions that hit the Bloom filter of the block header $BH_c$ to their original transaction fee, that is, $\forall i \in TH_c$, the expected value of transaction $i$ should be updated to $f_i / \epsilon$.\;
		}
		\BlankLine 
		
		\KwOn(Mine a new block $B_c$) { 
			Broadcast the block header $BH_c$ to all nodes in the blockchain network \;
			Broadcast the block body $BB_c$ to all nodes in the blockchain network \;
			Remove the transactions in $B_c$, i.e., $TB_c$, from the transaction pool. \;
		} 
		\BlankLine
		
		\KwOn(Receive a transaction $tx$) {
			Initiate the expected value of the transaction to be its  transaction fee\;
			Add the transaction into the transaction pool.\;
			\If{the transaction pool exceeds it limit}{
				Drop the transaction with the lowest expected value.
			}
		}
		
	\end{algorithm} 
	
	The first case happens when the miner receives a block header $BH_c$ (after the signal propagation delay $\tau$).
	If the block header $BH_c$ is valid, the miner should update the expected value of the transactions in $TH_c$. If a transaction $i$ hits the Bloom filter of a valid signal, its expected value should be multiplied by $\epsilon$, i.e., updated to $\epsilon f_i$, because the Bloom filter implies that the probability that the transaction $i$ is not included in the new block is only $\epsilon$.

	The second case happens when the miner receives a block body $BB_c$ (after the block propagation delay $\Delta$).
	If the block body is valid and matches the block header, the miner should remove the transactions that are included in the block body from the transaction pool.
	Then, the miner needs to recover the expected value of transactions due to false positive results from the Bloom filters, i.e., transactions that hit the Bloom filter but are not included in the block body. 
	Denote $U\backslash V=\{x\in U \text{ and } x\notin V\}$, then for all transactions $i \in  TH_c \backslash TB_c$, the expected value of transaction $i$ should be update to $f_i / \epsilon$. However, if the block body is invalid, we should update the expected values of all the transactions that hit the Bloom filter of the block header $BH_i$ by dividing the false positive rate, i.e., $\forall i \in TH_i$, the expected value of transaction $i$ should be updated to $f_i / \epsilon$.
	
	There exists a hard-coded timeout for the block header \cite{O1_propagation}. 
	When receiving a block header $BH_c$, a miner updates the expected transaction values accordingly, and starts a time counting. 
	If the miner has not received the corresponding block body until a certain timeout period, the block header will be marked as ``expired''. 
	Then, the miner will recover the expected values of all the transactions that hit the Bloom filter of the block header $BH_i$, that is, $\forall i \in TH_i$, the expected value of transaction $i$ should be updated to $f_i/ \epsilon$.
	
	To mine a new block, the miner will select a subset of transactions from the transaction pool based on certain transaction inclusion strategies, which is mentioned in Section \ref{sec:game} and will be further discussed in the following section.

	\subsection{Potential Cost of TIPS}
	
	In this section, we will discuss the potential cost of introducing TIPS in the DAG-based blockchain and show that the extra computation overhead and the communication overhead are insignificant.
	
	\textbf{Computation Overhead}:
	In TIPS, the extra computation overhead for the miners is from (i) maintaining the expected value for the transactions in the pool, and (ii) constructing and validating the Bloom filter in the block header. 
	We will then show that the computation overhead is negligible. 
	First, to maintain the expected value for the transactions, the miners can use some lookup data structure, such as the hash map, to update the transaction's expected value in $O(1)$. 
	Second, when mining a new block, the miners need to construct the Bloom filter in the block header. When receiving a newly-mined block, the miners also need to validate the Bloom filter, and this validation can be done by constructing the Bloom filter based on the transactions in the block body. 
	We next show that constructing the Bloom filter is computationally cheap.
	For a block including $2000$ transactions, the extra time for constructing the Bloom filter with $n=20000$ bits and $h=5$ hash functions, and validating it is only $6 \mu s$ on a PC with a CPU of Intel i7-1165G7. Therefore, we can claim that the extra computation overhead for introducing TIPS is insignificant.
	
	\textbf{Communication Overhead}:
	In TIPS, the miners need to broadcast the block header including the Bloom filter as the signal before broadcasting the block body, which will result in extra communication overhead. As shown in Section \ref{sec:bloom}, the total size of the signal is only $2.5$KB, which is less than $1\%$ of the original block size. 
	Besides, in practice, we can use blockchain relay networks to further accelerate the propagation of the block header \cite{bi2018accelerated}.
	Therefore, the extra communication overhead for introducing TIPS is also insignificant.
	
	\section{TIPS Breaks Down the Dilemmas}\label{sec:tips_break}
	
	In this section, we first the equilibrium  of the transaction inclusion game with TIPS and show that TIPS can lower the effective network delay and thus make the equilibrium transaction inclusion strategy approach to the top $n$ strategy. In this way, we show that TIPS can break the revenue and throughput dilemma.
	
	\subsection{Lowering Effective Network Delay}

	We first study the equilibrium  of the transaction inclusion game with TIPS.
	The following theorem shows an interesting result that when the false positive rate of the Bloom filter is low enough, the equilibrium is similar to the case in Theorem \ref{th:NE_general}, but with the effective network delay drastically decreased to $\tau$, where $\tau$ is the network propagation delay for the signal in TIPS.
	\begin{theorem}\label{th:NE_signal}
		If the false positive probability of the Bloom filter satisfies $\epsilon < \frac{f_{m}}{f_1}$, the symmetric equilibrium in TIPS is $\textbf{p}^*(\tau)$, where 
		\begin{equation}\nonumber
			\begin{aligned}
				p_i^*(\tau) = \begin{cases}
					\min \{r^{-1}\left( \frac{c_{l^*} }{f_i} | \tau \right), 1  \}, & 1 \leq i \leq l^*,\\
					0, & l^* < i \leq m.
				\end{cases} \\
			\end{aligned}
		\end{equation}
		Also, we have
		\begin{equation}\nonumber
			\begin{aligned}
				&F_l(c) = \sum_{i=1}^{l} \min \{ r^{-1}(c/f_i | \tau),1 \} - n, \quad \forall 1 \leq l \leq m, \\
				&l^* = \max \{ l \leq m | \forall i \leq l: F_i\left( f_i \right) \leq 0  \}, \\
				&c_{l^*} \text{ is the root of } F_{l^*}.
			\end{aligned}
		\end{equation}
	\end{theorem}
	
	\begin{proof}
		As shown in Figure \ref{fig:model}, once a miner successfully mines a new block including a set of transactions $\mathbb{T}$, he will broadcast the signal immediately. Other miners will receive the signal after the signaling propagation delay $\tau$. If other miners would not include any transaction in $\mathbb{T}$ after receiving the signal, then Theorem 2 can be proved. Consequently, to prove the above theorem, it is sufficient to show that when $\epsilon < \frac{f_{m}}{f_1}$, it is not rational for other miners to include any transaction in $\mathbb{T}$ after receiving the signal.
		
		We consider the opposite holds, that  is, other miners will include at least one of the transactions in $\mathbb{T}$ in their block with a a positive probability, even after they have received the corresponding signal but not received the whole block yet. In this situation, we denote the expected reward for a miner to include transaction $i$ given that other miners include transaction $i$ in their blocks with probability $p_i$ is $f_i \cdot \tilde{r}^*(p_i)$. Then we have $\tilde{r}^*(p) < r(p | \tau)$. Besides, $\tilde{r}^*(\cdot)$ is also monotonically decreasing, so is its inverse function $\tilde{r}^{*-1}(\cdot)$. 
		We analyze the situations before receiving the signal (Situation A) and after receiving the signal (Situation B) as follows.
		
		\textbf{Situation A: Before receiving the signal.}
		The transactions in the pool are sorted in descending order of their fees, i.e., $\mathbb{S}_A = \{ tx_1, tx_2, \ldots, tx_m  \}$ with the corresponding subscript indexs of $s_A = \{1,2, \ldots,m \}$. According to Theorem 1, we know that $\textbf{p}^A$ is a symmetric equilibrium, where
		\begin{equation}\nonumber
			\begin{aligned}
				p^A_i = \begin{cases}
					\min \{\tilde{r}^{-1}\left( \frac{c_{l^*} }{f_{s_A(i)}} \right), 1  \}, & 1 \leq i \leq l_A^*,\\
					0, & l_A^* < i \leq m,
				\end{cases} \\
			\end{aligned}
		\end{equation}
		and we have
		\begin{equation}\nonumber
			\begin{aligned}
				&F^A_{l}(c) = \sum_{i=1}^{l} \min \{ \tilde{r}^{-1}(c/f_{s_A(i)}),1 \} - n, \quad \forall 1 \leq l \leq m, \\
				&l_A^* = \max \{ l \leq m | \forall i \leq l: F^A_i\left( f_{s_A(i)} \right) \leq 0  \}, \\
				&c_{l_A^*} \text{ is the root of } F^A_{l^*}.
			\end{aligned}
		\end{equation}
		Therefore, a transaction $i$ will be included only when $i \leq l_A^*$.

		\textbf{Situation B: After receiving the signal.}
		The corresponding signal will be broadcast to the whole network. 
		Without loss of generality, we consider transaction $tx_{\tilde{i}}$, where $\tilde{i} \leq l_A^*$. After the miner includes $tx_{\tilde{i}}$ in its block and  the corresponding signal has been received by other miners, the expected value of $tx_{\tilde{i}}$, i.e., $e_{\tilde{i}}$, will be $\epsilon f_{\tilde{i}}$. Since $\epsilon < f_{m} / f_1$, we have
		\begin{equation}\nonumber
			\epsilon f_{\tilde{i}} \leq \epsilon f_1 < f_{m} \leq f_{l_A^* + 2}.
		\end{equation}
		Then the transactions in the pool  sorted in descending order of their expected values will be 
		\begin{equation}\nonumber
			\mathbb{S}_B = \{ tx_1, \cdot\cdot, tx_{\tilde{i}-1}, tx_{\tilde{i}+1}, \cdot\cdot, tx_{l^*}, tx_{l^*+1}, \cdot\cdot, tx_{\tilde{i}}, \cdot\cdot, tx_m  \},
		\end{equation}
		with the corresponding subscript indexs of 
		\begin{equation}\nonumber
			s_B = \{1,2, \cdot\cdot, \tilde{i}-1,\tilde{i}+1,\cdot\cdot, l_A^*,l_A^*+1, \cdot\cdot, \tilde{i},\cdot\cdot, m \}.
		\end{equation}
		According to Theorem 1, we know that $\textbf{p}^B$ is a symmetric equilibrium, where
		\begin{equation}\nonumber
			\begin{aligned}
				p^B_i = \begin{cases}
					\min \{\tilde{r}^{-1}\left( \frac{c_{l^*} }{f_{s_B(i)}} \right), 1  \}, & 1 \leq i \leq l^*,\\
					0, & l^* < i \leq m,
				\end{cases} \\
			\end{aligned}
		\end{equation}
		and we have
		\begin{equation}\nonumber
			\begin{aligned}
				&F^B_{l}(c) = \sum_{i=1}^{l} \min \{ \tilde{r}^{-1}(c/f_{s_B(i)}),1 \} - n, \quad \forall 1 \leq l \leq m, \\
				&l_B^* = \max \{ l \leq m | \forall i \leq l: F_i\left( f_i \right) \leq 0  \}, \\
				&c_{l_B^*} \text{ is the root of } F^A_{l^*}.
			\end{aligned}
		\end{equation}
		Then we have
		\begin{equation}\nonumber
			\begin{aligned}
				&F^B_{l_A^* + 1}(f_{s_B(l_A^*+1)}) - F^A_{l_A^* + 1}(f_{s_A(l^*+1)}) \\
				=& \sum_{i=1}^{l_A^* + 1} \min \{ \tilde{r}^{-1}\left( \frac{f_{l_A^* + 1}}{f_{s_A(i)}}  \right), 1   \} - \!\! \sum_{i=1}^{l_B^* + 1} \min \{ \tilde{r}^{-1}\left( \frac{f_{l_B^* + 1}}{f_{s_B(i)}}  \right), 1   \} \\
				=& \min \{ \tilde{r}^{-1}\left( \frac{f_{l_A^* + 2}}{f_{l_A^* + 2}}  \right), 1   \} - \min \{ r^{-1}\left( \frac{f_{l_A^* + 1}}{f_{\tilde{i}}}  \right), 1   \} \\
				=& 1 - \min \{ \tilde{r}^{-1}\left( \frac{f_{l_A^* + 1}}{f_{\tilde{i}}}  \right), 1   \} \geq 0
			\end{aligned}
		\end{equation}
		Therefore, we have $F^B_{l_A^* + 1}(f_{s_B(l_A^*+1)}) > 0$, which implies that
		\begin{equation}\nonumber
			\text{Index}(tx_{\tilde{i}}) > l_A^* + 1 > l_B^*.
		\end{equation}
		Thus, transaction $tx_{\tilde{i}}$ would not be included in any block after receiving the corresponding signal, which contradicts the previous assumption. The proof is thus completed.
	\end{proof}
	
	\ADD{In the condition $\epsilon < f_m / f_1$ in Theorem \ref{th:NE_signal}, $\epsilon$ denotes the false positive probability of the Bloom filter, $f_1$ denotes the highest transaction fee, and $f_m$ denotes the lowest transaction fee. This condition implies that when the spread of the transaction fee is large, we need a Bloom filter with high precision, i.e., low false positive probability, to achieve the equilibrium strategy in Theorem \ref{th:NE_signal}. Intuitively, when $\epsilon < f_m / f_1$, once a transaction $i$ with fee $f_i (f_m \leq f_i \leq f_1)$ hits the Bloom filter, its expected value will be updated to $\epsilon f_i$, where $\epsilon f_i < f_m$. In this case, other miners would avoid including transaction $i$ in their newly-mined block, and thus the equilibrium strategy in Theorem \ref{th:NE_signal} holds.}
	
	\noindent\textbf{Remark: TIPS shortens effective network delay.}
	Compared with the equilibrium in Theorem \ref{th:NE_general}, we find that in TIPS, after the time period $\tau$ of the signal propagation, the transactions included in the newly-mined block will not be included by any rational miner. Therefore, TIPS shortens the effective network propagation delay from the block broadcast time $\Delta$ to the drastically smaller signal broadcast delay $\tau (\tau \ll \Delta)$, because as mentioned in Section \ref{sec:bloom}, the size of the signal can be smaller than $1/100$ the size of the whole block. 
	The equilibrium in Theorem \ref{th:NE_signal} holds in TIPS when the Bloom filter's false positive is small enough, i.e., $\epsilon < f_m / f_1$.  This implies that when the Bloom filter's false positive is small enough, the signal in TIPS can provide information which is precise enough for miners to avoid transaction collisions.
	We will next show that this condition on the Bloom filter is not restrictive in real-world blockchain systems. 
	
	
	\begin{figure}[!ht]
		\includegraphics[width=0.9\linewidth,height=0.3\linewidth]{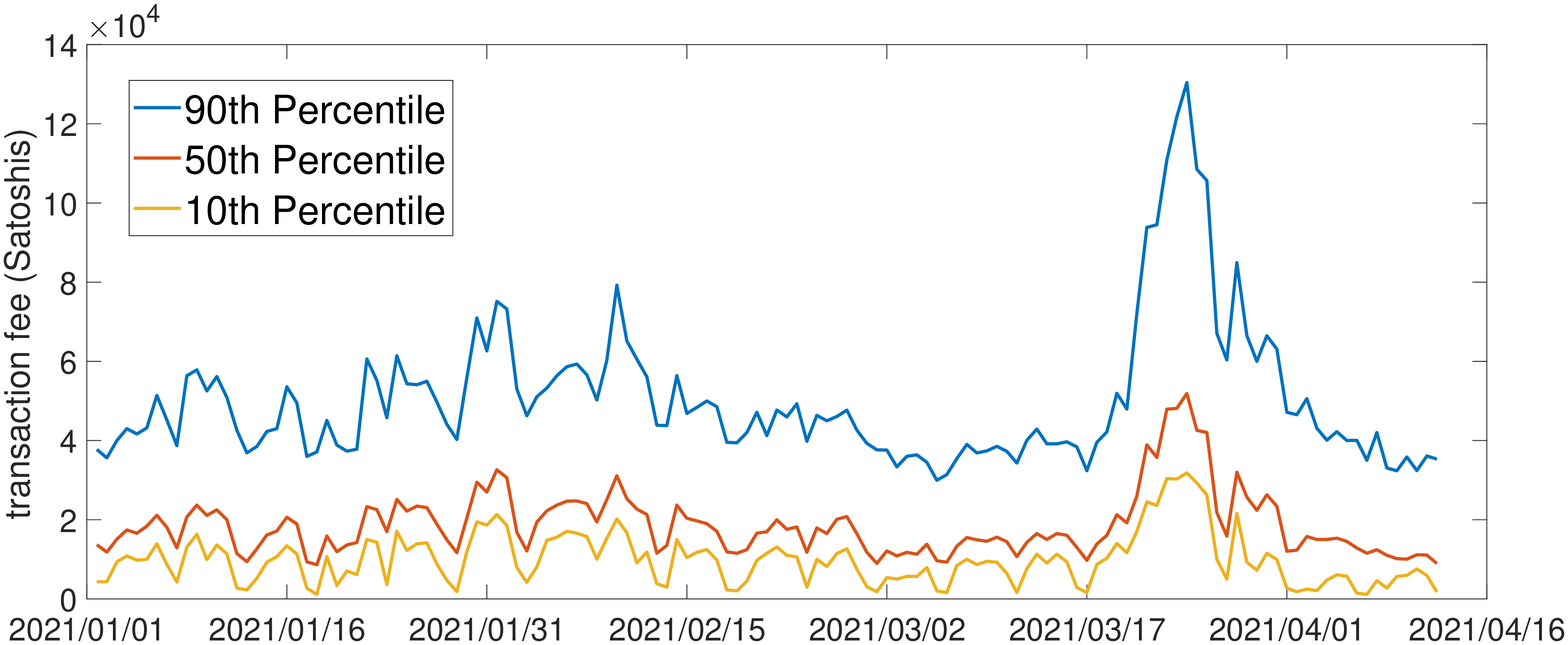}
		\caption{Distribution of transaction fees in Bitcoin}
		\label{fig:fees_distribution}
	\end{figure}

	To see that the condition of $\epsilon<f_{m}/f_{1}$ is practical, we collect the data from Bitcoin \cite{bitcoinvisuals} and investigate the distribution of the transaction fees. Bitcoin admits one of the most popular blockchain transaction markets, thus, we believe that the transaction fee statistics in Bitcoin is representative enough.
	We find that in realistic blockchain systems like Bitcoin, there exist few transactions with unnecessarily high transaction fees and few transactions with extremely low transaction fees. To avoid the impact of this abnormal transaction fee, we focus on the 90th percentile and the 10th percentile for transaction fees of the transactions included in the block, which is visualized in Figure \ref{fig:fees_distribution}. We have that
	\begin{equation}\nonumber
		\omega = \min \frac{\text{10th percentile}}{\text{90th percentile}} = 0.0304 .
	\end{equation}
	As mentioned in Section \ref{sec:bloom}, if we allocate 8 bits in the Bloom filter for a transaction in the pool, i.e., $b/n = 8$, the false positive probability $\epsilon = 0.0217$, which is smaller than $\omega$ defined above. This implies that the condition that the false positive probability of the Bloom filter should be small enough, i.e., $\epsilon < f_{m} / f_1$, in Theorem \ref{th:NE_signal}, is almost always satisfied in real-world blockchain systems. 
	Hence, we consider a scenario where the condition $\epsilon < f_{m} / f_1$ is always satisfied in the following discussions.


	\subsection{Approaching Top $n$ Strategy}\label{sec:approach:top_n}
	
	Here, we will show that by lowering the effective network delay, the equilibrium transaction inclusion strategy in TIPS approaches to the top $n$ strategy, which makes it possible for TIPS to break down the dilemmas in DAG-based blockchain.
	
	We first show that the top $n$ strategy in TIPS is an $\eta$-approximate Nash equilibrium strategy, and $\eta$ represents the gap between the top $n$ strategy and the equilibrium strategy.
	
	\begin{theorem}\label{th:approximate}
		The top n strategy, i.e., $\textbf{p}^{\rm top}$ is an $\eta$-approximate Nash equilibrium, where
		\begin{equation}\nonumber
			\eta \leq \left| n \left(1 - \frac{1 - e^{-\lambda \tau}}{\lambda \tau}  \right) f_n \right|.
		\end{equation}
		The equation holds if and only if the transactions are homogeneous, that is, the transaction fees are the same.
		Specially, when $\tau \rightarrow 0$, the top $n$ strategy is the Nash equilibrium.
	\end{theorem}

	For convenience, we denote that $g(\tau) = 1 - \frac{1 - e^{-\lambda \tau}}{\lambda \tau}$. Note that $g(\tau)$ is monotonically increasing in $\tau$, which implies that the extra revenue from deviating from the top $n$ strategy will decrease as the effective propagation delay becomes smaller and the equilibrium transaction inclusion strategy will gradually converge to the top $n$ strategy. The effective network propagation delay in TIPS is $\tau$, which is drastically small enough. Therefore, the top $n$ strategy in TIPS is a good approximation equilibrium strategy. Besides, we show that when the effective network propagation delay $\tau$ is smaller enough,  the top $n$ strategy in TIPS can be the unique equilibrium strategy, which is shown in the following theorem.
	

	\begin{theorem}\label{th:NE_top_n}
		The top $n$ strategy, i.e., $\textbf{p}^{\rm top}$ is the unique Nash equilibrium when
		\begin{equation}\label{eq:NE_condiction}
			\tau \leq \frac{1}{\lambda} \varphi^{-1}\left( \frac{f_{n+1}}{f_{n}} \right),
		\end{equation}
		where $\varphi(x) = \frac{1 - e^{-x}}{x}$, and $\varphi^{-1}(x)$ is its inverse function.
	\end{theorem}

	Theorem \ref{th:approximate} and Theorem \ref{th:NE_top_n} show that the top $n$ strategy is very possible to be the equilibrium strategy in TIPS, and is at least a good approximation transaction inclusion strategy, because the effective network propagation delay in TIPS  $\tau$  is drastically small enough since we only need to broadcast a small-size signal. Therefore, we adopt the top $n$ strategy in TIPS in the following analysis.

	\subsection{Breaking Down the Revenue Dilemma}

	The analysis above shows that breaking down the revenue dilemma, i.e., achieving the top $n$ strategy, is difficult, because the top $n$ strategy is the Nash equilibrium strategy only if the network delay is very small. This is not practical in the standard protocol (i.e., without TIPS) since broadcasting the whole block in the blockchain network takes a long time. However, TIPS drastically lowers the effective network propagation delay, making the top $n$ strategy feasible, which breaks down the revenue dilemma.
	
	Figure \ref{fig:equilibrium_strategy} shows the comparison of equilibrium strategies with TIPS and the standard protocol, under a moderate parameter $\tau=0.01 \Delta$, while the Bloom filter is usually much less than $1\%$ of the size of the block body (See Section \ref{sec:experiment} for detailed experiment set up). Figure \ref{fig:equilibrium_strategy} greatly supports that the top $n$ strategy is a good approximation transaction inclusion strategy in TIPS, while the equilibrium strategy in the standard protocol (i.e., without TIPS) gradually converges to the random strategy when the network propagation delay is large.
	
	\begin{figure}[!ht]
		\includegraphics[width=1.0\linewidth,height=0.4\linewidth]{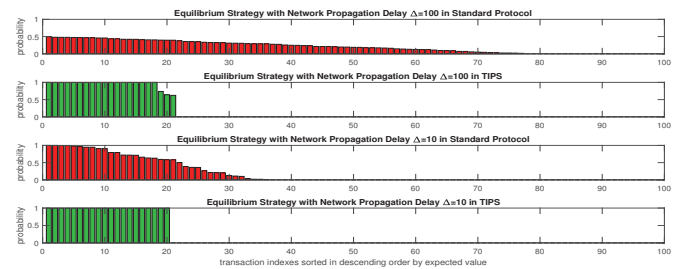}
		\caption{Equilibrium strategy with different network propagation delay}
		\label{fig:equilibrium_strategy}
	\end{figure}
	
	To demonstrate the efficiency of TIPS under the revenue dilemma, we investigate the performance of the equilibrium strategy under TIPS as follows.
	\begin{definition}
		The efficiency of the  equilibrium strategy of the transaction inclusion game in the DAG-based blockchain under the revenue dilemma is defined as the ratio of the miners' revenue $R$ under the equilibrium strategy of the transaction inclusion game and the highest miners' revenue achieved by any transaction inclusion strategy, which is shown as follows:
		\begin{equation}
			\rm Efficiency(R) = \frac{\rm Revenue \ of \ Equilibrium}{\rm Revenue \  of \ optimal \ strategy}.
		\end{equation}
	\end{definition}
	
	Then we have the following result, which can demonstrate the high efficiency of TIPS.
	\begin{theorem}\label{th:efficiency:fsr}
		The efficiency of the equilibrium strategy of the transaction inclusion game in the DAG-based blockchain with TIPS under the revenue dilemma is
		\begin{equation}\nonumber
			\begin{aligned}
				&  \text{Efficiency}(R) \geq  \frac{(1 - e^{-\lambda \tau})}{\left( \lfloor \frac{m}{n} \rfloor - \sum_{k=0}^{\lfloor \frac{m}{n} \rfloor} \frac{(\lambda \tau)^k}{k!} e^{-\lambda \tau} \right)} .
			\end{aligned}
		\end{equation}
	\end{theorem} 
	According to Theorem \ref{th:efficiency:fsr}, we have that $\lim_{\tau \rightarrow 0 } \text{Efficiency}(R) = 1$. 
	\ADD{This is because when $\tau \rightarrow 0$, the top $n$ strategy $\textbf{p}^{\text{top}}$ is the unique Nash equilibrium as shown in Theorem \ref{th:NE_top_n}. Besides, based on Lemma \ref{lemma:reward}, we have $\lim_{\Delta=\tau \rightarrow 0} R(\textbf{p}^{\text{top}}) = \sum_{i=1}^{n} f_i$, which implies that the miner can also obtain the highest transaction fee reward.}
	Thus TIPS can achieve a near-optimal miners' revenue, and therefore can efficiently break down the revenue dilemma.
	

	\subsection{Breaking Down the Throughput Dilemma}
	
	TIPS breaks the throughput dilemma by avoiding collision through quickly delivering a small signal to the network. 
	As mentioned in Section \ref{sec:bloom}, the size of the Bloom filter is much smaller compared to the block size, and it is independent to the transaction size.	
	In this case, we can increase the block size, while the broadcast time of the signal remains almost the same.
	\ADD{Since TIPS greatly shortens the effective network delay, according to Theorem \ref{th_tps}, a smaller network effective delay $\Delta$ contributes to a high block capacity utilization and thus a high system throughput.}
	This makes it possible to contain more transactions in the block (i.e., a larger $n$) without significantly sacrificing the block capacity utilization.
	Thus, TIPS can break down the throughput dilemma.

	From Theorem \ref{th_tps}, we know that the system throughput increases with a large block size $n$. To further observe the impact of TIPS on the system throughput, we consider the limit of throughput, that is, the maximum throughput that the system can achieve by increasing the block size, i.e., $\lim_{n \rightarrow \infty} TPS(\textbf{p}, n)$. We denote $\Delta(n)$ as the network propagation delay for a block of size $n$. As discussed in Section \ref{sec:approach:top_n}, with the introduction of TIPS, the equilibrium transaction inclusion strategy approach to the top $n$ strategy, i.e., $\textbf{p}^{\text{top}}$. Thus, we consider the limit of throughput when the miners take the top $n$ transaction inclusion strategy in the following lemma.
	\begin{lemma}\label{lemma:limit_tps}
		The limit of throughput of the DAG-based blockchain with the top $n$ transaction inclusion strategy $\textbf{p}^{\text{top}}$ is 
		\begin{equation}\nonumber
			\lim_{n \rightarrow \infty} TPS(\textbf{p}^{\text{top}}, n) = \lim_{n \rightarrow \infty} \frac{1}{\frac{\mathrm{d}\Delta(n)}{\mathrm{d}n}},
		\end{equation}
		where $\frac{\mathrm{d}\Delta(n)}{\mathrm{d}n}$ is the derivative of $\Delta(n)$.
	\end{lemma}
	In Lemma \ref{lemma:limit_tps}, $\frac{\mathrm{d}\Delta(n)}{\mathrm{d}n}$ denotes the marginal network propagation time due to an extra transaction in the block.
	Empirical results show that the marginal network propagation time for each KB of data is a constant, e.g., 80ms for each extra KB of Bitcoin block \cite{decker2013information}.
	We compared the limit of throughput of TIPS and the standard protocol below.
	\begin{equation}\label{eq:limit:tps}
		\begin{aligned}
			\frac{\lim_{n \rightarrow \infty} TPS^{\text{TIPS}}(\textbf{p}^{\text{top}}, n)}{\lim_{n \rightarrow \infty} TPS^{\text{Standard}}(\textbf{p}^{\text{top}}, n)}&=\!\frac{\lim_{n \rightarrow \infty} \frac{1}{\frac{\mathrm{d}\tau(n)}{\mathrm{d}n}}}{\lim_{n \rightarrow \infty} \frac{1}{\frac{\mathrm{d}\Delta(n)}{\mathrm{d}n}}}\!\! =\!\! \lim_{n \rightarrow \infty}\frac{\frac{\mathrm{d}\Delta(n)}{\mathrm{d}n}}{\frac{\mathrm{d}\tau(n)}{\mathrm{d}n}} \\
			&= \frac{\text{transaction size}}{\varsigma\text{ bits}},
		\end{aligned}
	\end{equation}
	where $\varsigma$ denotes the number of bits per transaction in the Bloom filter. In (\ref{eq:limit:tps}), we see that the throughput ratio between the standard protocol and TIPS is essentially the ratio between the size of the needed broadcast for a marginal transaction. 
	Thus, the standard protocol is limited by the broadcast of the additional transaction, while TIPS is limited only by the broadcast of its signal (with a much smaller data size). Practically, a Bloom filter with $\varsigma=b/n = 8$ can achieve a false positive probability of $2.17\%$, which is practical enough in TIPS. While the average transaction size in Bitcoin is 500 KB. Therefore, theoretically, the ratio of the limit of throughput in TIPS and the standard protocol can be as large as $5\times 10^5$.
	This huge ratio shows the potential of TIPS in breaking the throughput dilemma.
	
		\begin{figure}[!th]
		\centering
		\subfigure[TPS in TIPS and Standard Protocol]{
			\includegraphics[width=2.62in,height=1.3in]{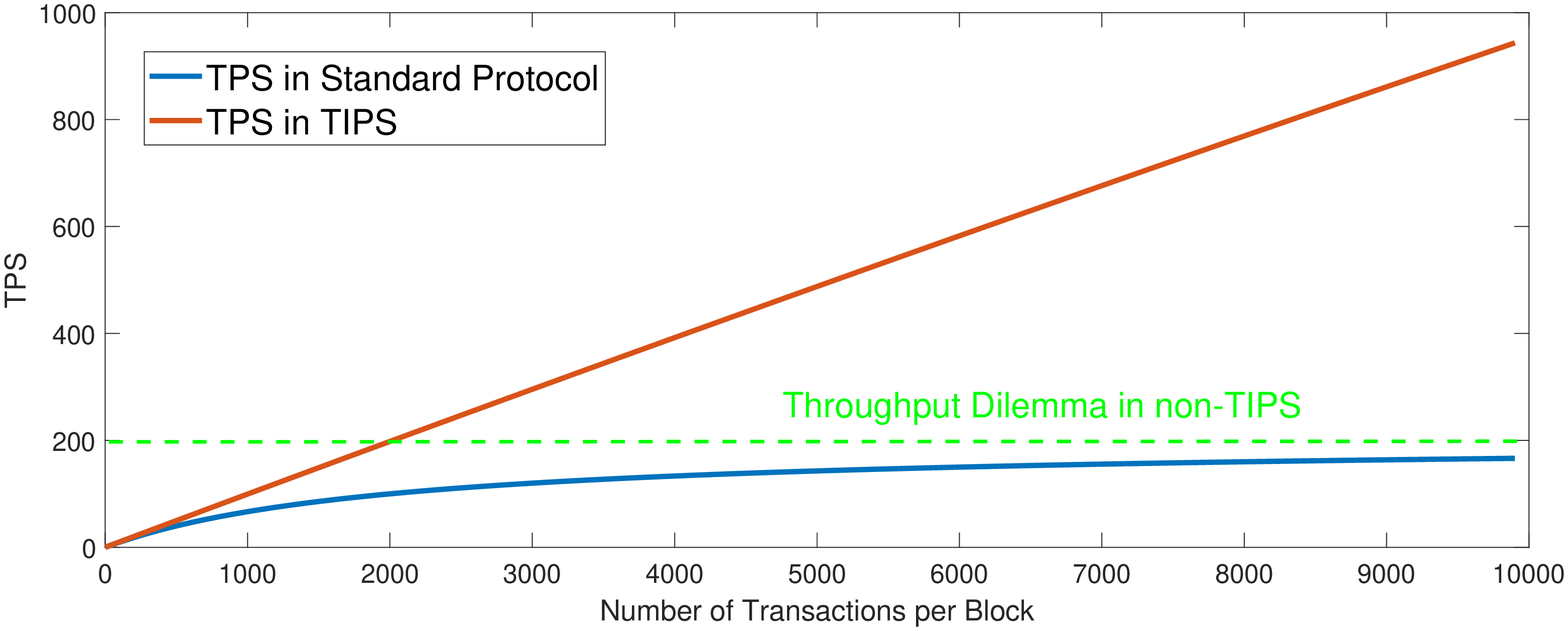}
		}
		\subfigure[Utilization in TIPS and Standard Protocol]{
			\includegraphics[width=2.62in,height=1.3in]{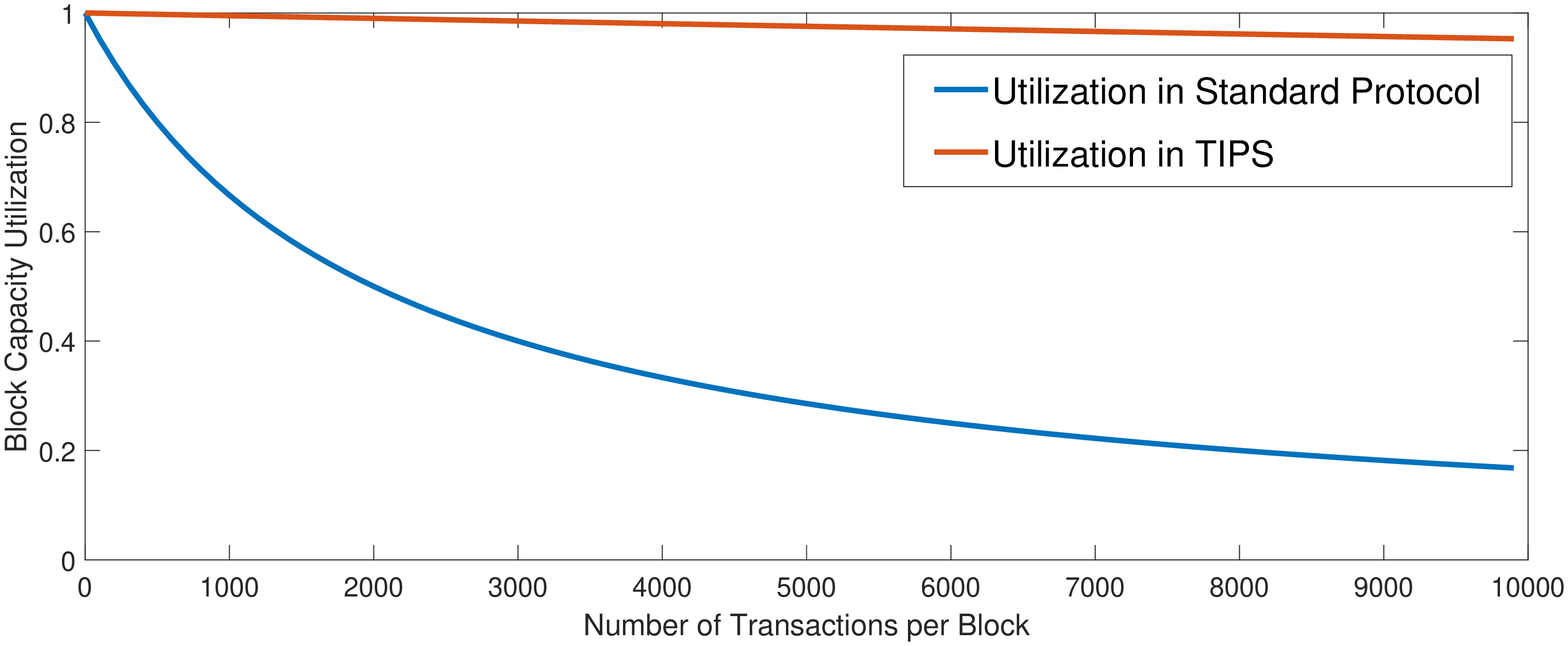}
		}
		\caption{Breaking down the throughput dilemma with TIPS}
		\label{fig:tps_dilemma}
	\end{figure}
	
	\ADD{Figure \ref{fig:tps_dilemma} shows that increasing the block size will drastically reduce the block capacity utilization due to the increasing transaction inclusion collision and that the low utilization further limits the system throughput. 
	However, by broadcasting the almost accurate signal from the most recent header in a short time, TIPS efficiently avoids transaction inclusion collision, contributing to a high utilization, and thus greatly boosts the system throughput.
	}
	
	To further demonstrate the efficiency of TIPS under the throughput dilemma, we investigate the performance of the equilibrium strategy of the transaction inclusion game with TIPS as follows.
	\begin{definition}
		The efficiency of the  equilibrium strategy of the transaction inclusion game in the DAG-based blockchain under the throughput dilemma is defined as the ratio of the TPS under the equilibrium strategy of the transaction inclusion game and the highest TPS achieved by any transaction inclusion strategy, which is shown as follows:
		\begin{equation}
			\rm Efficiency(TPS) = \frac{\rm TPS \ of \ Equilibrium}{\rm TPS \  of \ optimal \ strategy} .
		\end{equation}
	\end{definition}
	
	Then we have the following results.
	\begin{theorem}\label{th:efficiency:tps}
		The efficiency of the equilibrium strategy of the transaction inclusion game in the DAG-based blockchain with TIPS under the throughput dilemma is
		\begin{equation}\nonumber
			\begin{aligned}
				&  \text{Efficiency}(TPS) \geq  \frac{(1 - e^{-\lambda\tau})}{\left( \lfloor \frac{m}{n} \rfloor - \sum_{k=0}^{\lfloor \frac{m}{n} \rfloor} \frac{(\lambda \tau)^k}{k!} e^{-\lambda\tau} \right)} .
			\end{aligned}
		\end{equation}
		
	\end{theorem} 
	\ADD{According to Theorem \ref{th:efficiency:tps}, we have that $\lim_{\tau \rightarrow 0 } \text{Efficiency}(TPS) = 1$, 
		which is consistent with Theorem \ref{th_tps}. In Theorem \ref{th_tps}, we have $\lim_{\Delta=\tau \rightarrow 0} U(\textbf{p}) = 1$, which implies that TIPS can achieve an extremely high block capacity utilization.}
	Thus, TIPS can achieve near-optimal TPS, and therefore can efficiently break down the throughput dilemma.

	\section{Experiment Results}\label{sec:experiment}
	
	
	In this section, we conduct experiments to demonstrate the performance of TIPS and validate our analysis.
	
	\ADD{We develop a DAG-based blockchain simulator in Python using SimPy \cite{matloff2008introduction}. We implement the basic inclusive protocol \cite{lewenberg2015inclusive} in the simulator, which is one of the most famous DAG-based blockchain protocols.
	The implementation is sufficiently representative because TIPS is a robust ``add-on'' design, which can be applied to most of the current DAG-based blockchain protocols, such as Conflux \cite{li2018scaling} and CDAG \cite{gupta2019cdag}. We have realized all designed protocols in TIPS in the simulator, including the block propagation model, the construction and validation of the Bloom filter, the operation of maintaining the expected value of the transactions in the transaction pool, and different transaction inclusion strategies. We use a Python package ``pybloom-live'' \cite{onlineBloom} as the implementation of the Bloom filter data structure in the TIPS. There are 10 homogeneous miners in the simulators, which are connected to a P2P network. Each miner will maintain a local version of the transaction pool and will select some transactions in the pool based on the given transaction inclusion strategy. With this simulator, we will compare the system performance of the standard inclusive protocol (without TIPS), and protocol with TIPS under different transaction inclusion strategies.}
	
			\begin{figure*}[!t]
		\centering
		\begin{minipage}[t]{0.32\linewidth}
			\centering
			\includegraphics[width=0.95\linewidth,height=0.7\linewidth]{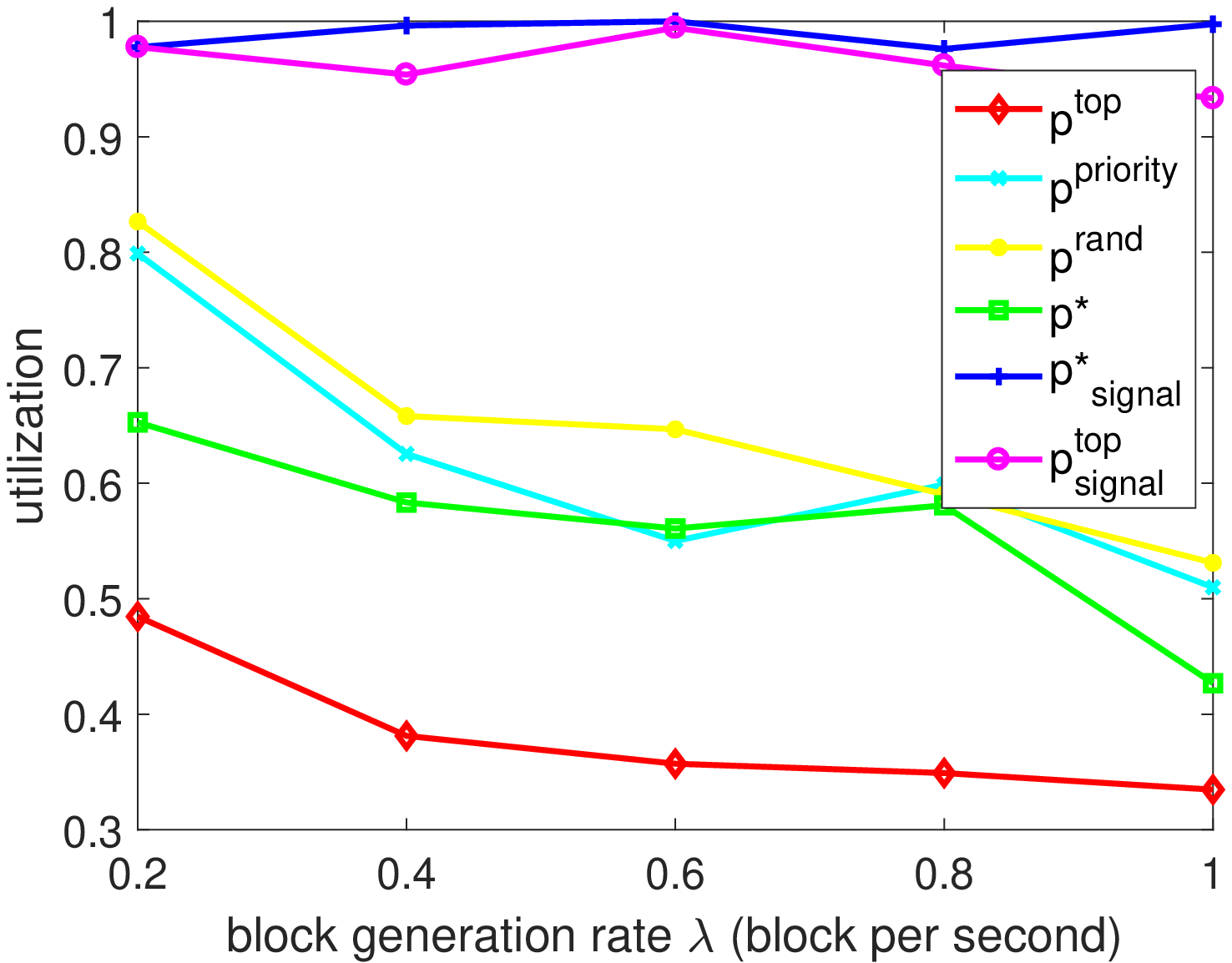}
			\caption{Utilization of different transaction inclusion protocols}
			\label{fig:exp.performance.utility}
		\end{minipage}
		\hfill
		\begin{minipage}[t]{0.32\linewidth}
			\centering
			\includegraphics[width=0.95\linewidth,height=0.7\linewidth]{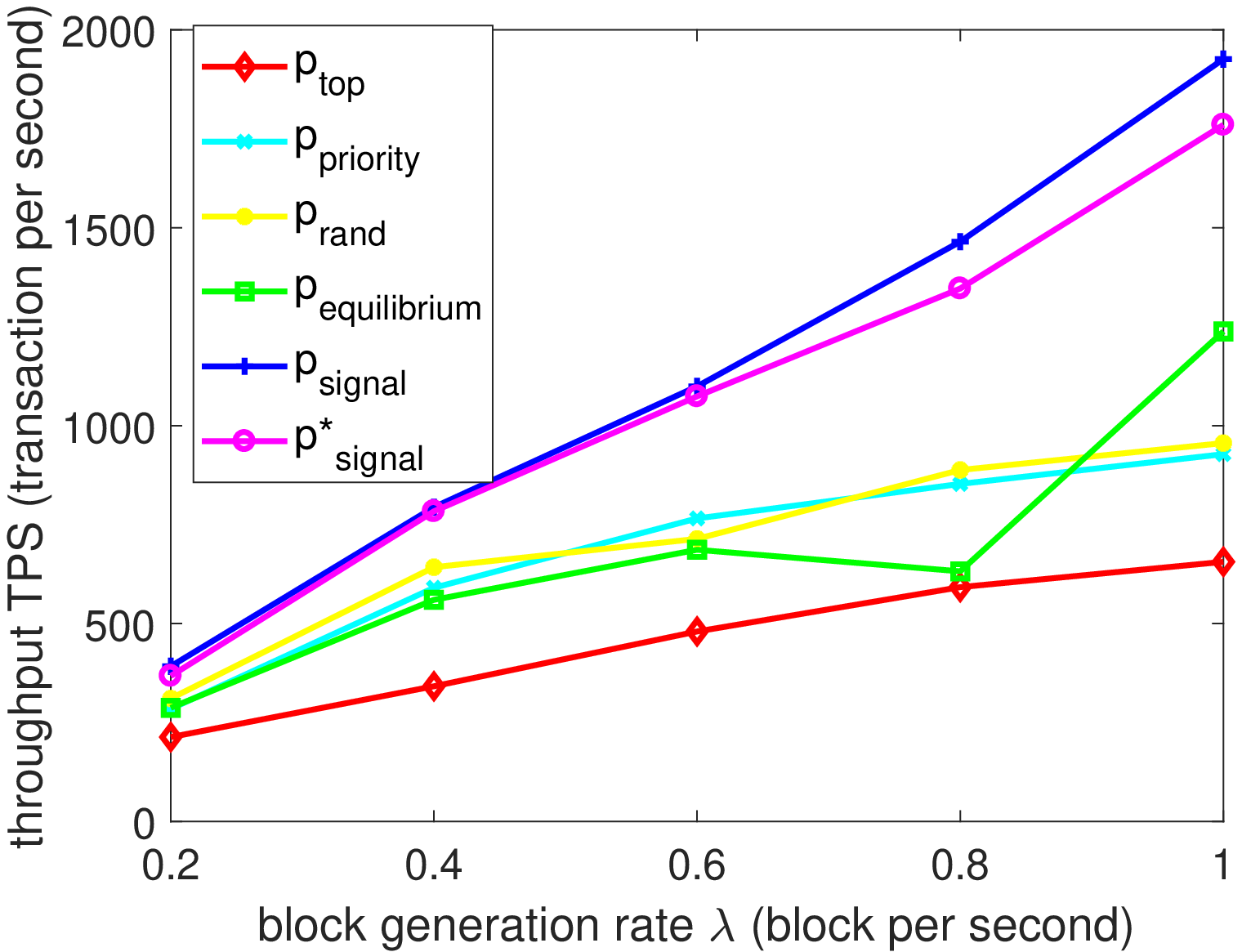}
			\caption{TPS of different transaction inclusion protocols}
			\label{fig:exp.performance.tps}
		\end{minipage}
		\hfill
		\begin{minipage}[t]{0.32\linewidth}
			\centering
			\includegraphics[width=0.95\linewidth,height=0.7\linewidth]{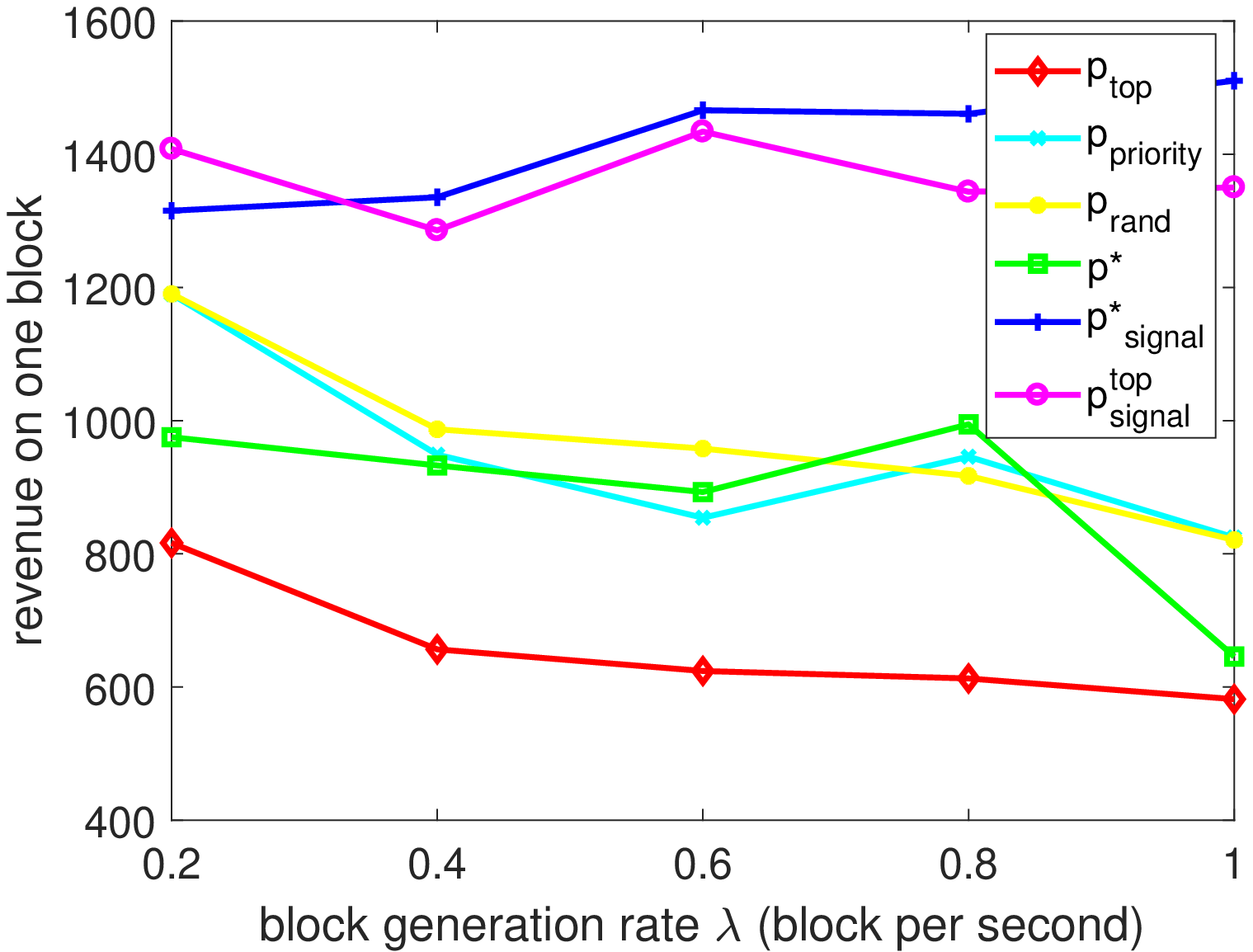}
			\caption{Miners' revenue of different transaction inclusion protocols}
			\label{fig:exp.performance.fsr}
		\end{minipage}
	\end{figure*}

	\begin{figure*}[!t]
		\centering
		\begin{minipage}[t]{0.48\linewidth}
			\centering
			\includegraphics[width=3.3in,height=1.2in]{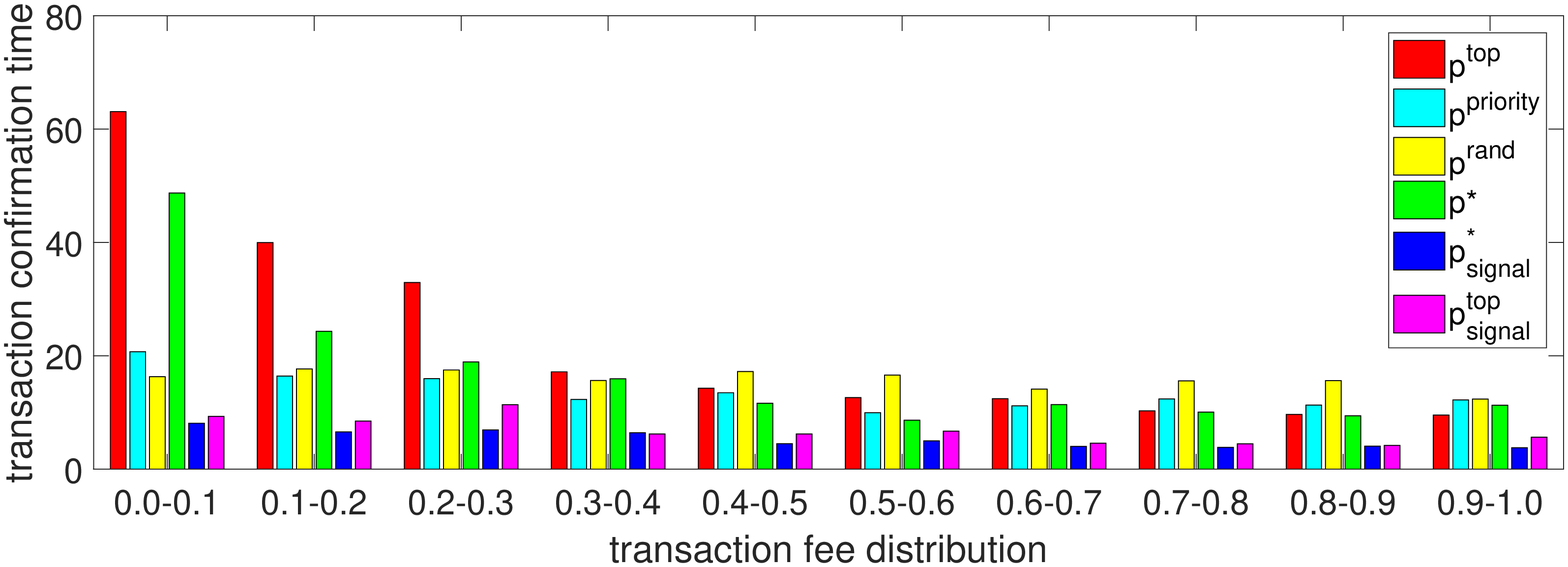}
			\caption{Average transaction confirmation time under different transaction inclusion protocols} 
			\label{fig:exp.performance.waiting_time}
		\end{minipage}
		\hfill
		\begin{minipage}[t]{0.48\linewidth}
			\includegraphics[width=3.3in,height=1.2in]{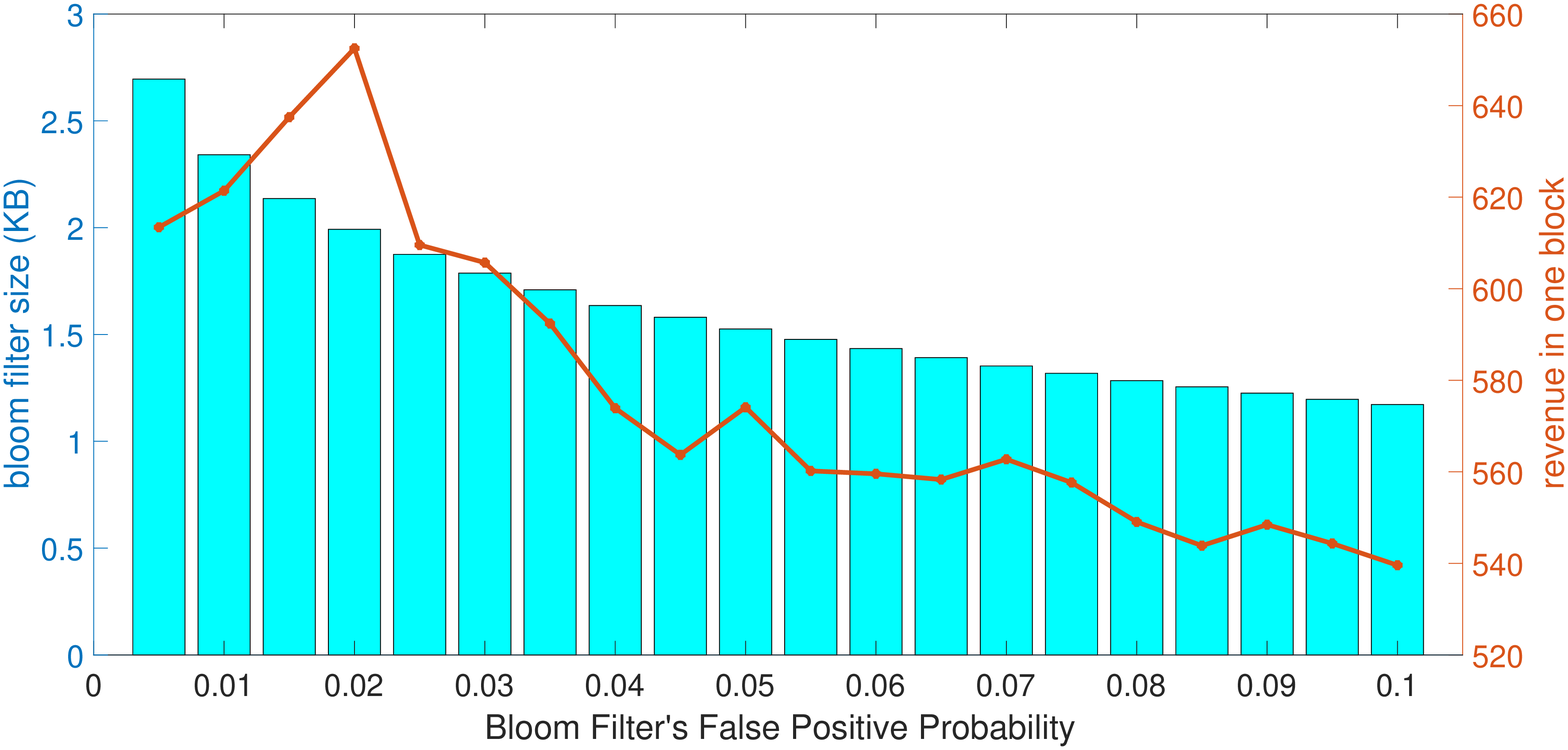}
			\caption{Experimental results with different Bloom filters' false positive probabilities}
			\label{fig:bf_size_fsr}
		\end{minipage}
	\end{figure*}
	The experimental configuration is as follows.
	We set the block size to be 1MB, which is the current block size limitation in Bitcoin. With the average transaction size being 500 bytes, we put 2000 transactions in one block, i.e., $n = 2000$. Besides, we assume the size of the transaction pool to be $m = 10000$. The propagation delay for the whole block is a random variable following the normal distribution with the expectation as $\Delta = 10$, and the propagation delay for the signal is a random variable following the normal distribution with the expectation as $\tau = 0.1$. The block generation rate of the DAG-based blockchain system $\lambda$ ranges from 0.1 to 1. The arrival of the transaction follows the Poisson process with the transaction fee drawn uniformly from $[0,1]$.
	
	In Figure \ref{fig:exp.performance.utility} to Figure \ref{fig:exp.performance.waiting_time}, results of $p^{\rm top}$, $p^{\rm priority}$, $p^{\rm rand}$ and $p^*$ correspond to the cases when all miners adopt strategies $\textbf{p}^{\rm top}$,$\textbf{p}^{\rm priority}$, $\textbf{p}^{\rm rand}$ and the equilibrium strategy in the standard protocol, respectively. 
	Meanwhile, results of $p^{*}_{\rm signal}$ and $p^{\rm top}_{\rm signal}$ represent the cases when all miners adopt the equilibrium strategy and the $\textbf{p}^{\rm top}$ strategy in TIPS, respectively.
	
	From Figure \ref{fig:exp.performance.utility}, we can find that TIPS always achieves high utilization and the utilization in TIPS is not sensitive to the increase of the block generation rate. 
	\ADD{This is because in the standard protocol, a higher block generation rate means that more blocks will be generated during the network propagation delay $\Delta$, leading to more transaction inclusion collisions and lower block capacity utilization. This phenomenon is also consistent with Theorem 3, where the utilization $U$ is monotonically decreasing with the generation rate $\lambda$ when other parameters are fixed. On the contrary, since TIPS drastically lowers the effective network propagation delay, it efficiently avoids the transaction inclusion collision and greatly improves the utilization.}

	Figure \ref{fig:exp.performance.tps} shows the system throughput (TPS) of different transaction inclusion protocols, from which we find that the throughput increases with the block generation rate.
	\ADD{This is consistent with Theorem \ref{th_tps}, i.e., the system throughput (TPS) is monotonically increasing with the block generation rate $\lambda$. }
	Besides, from Figure \ref{fig:exp.performance.tps}, we can also find that TIPS can always achieve a higher system throughput compared to the standard protocol, which demonstrates that TIPS effectively breaks down the throughput dilemma. \ADD{This is because TIPS greatly improves the block capacity utilization as shown in Figure \ref{fig:exp.performance.utility}, and therefore can process more distinct transactions than the standard protocol.}

	Figure \ref{fig:exp.performance.fsr} shows that TIPS can achieve a significantly higher miners' revenue compared to the standard protocol.
	\ADD{This is because transaction inclusion collision will split the transaction fee reward and thus lower the miners' revenue. Besides, the miners' collision-avoiding transaction inclusion strategies in the standard protocol tend to include cheap transactions, which also hurts miners' revenue. In the meantime, TIPS avoids the transaction inclusion collision and encourages miners to include the transactions with the highest fees, which improves the miners' revenue.} 
	This validates the claim that TIPS can effectively break down the revenue dilemma.
	
	
	Figure \ref{fig:exp.performance.waiting_time} shows the average transaction confirmation time for transactions with different transaction fees. Firstly, the average transaction confirmation time in  TIPS is significantly shorter than that of other protocols. 
	\ADD{This is because TIPS can achieve a higher system throughput, which means that more transactions can be processed in a short time.}
	Secondly,  TIPS guarantees that the transactions with higher transaction fees can be  confirmed in a shorter time,
	\ADD{since TIPS encourages the miners to include the transactions with the highest fees. The above two properties of TIPS together guarantee a good user experience.}

	Besides, we further investigate the impact of the Bloom filter's false positive probability, which is shown in Figure \ref{fig:bf_size_fsr}. We can find that a small Bloom filter has a high probability of false positive, and can be broadcast to the network in a short time, but will lead to the following drawbacks: (1) When the false positive probability is high, there are a number of transactions that hit the Bloom filter but are not included in the block. These transactions need to wait for a long time to be included in a block, so as some transactions with high transaction fees, which will degrade the system performance with lower miners' revenue; (2) When the false positive probability is high, the expected value of the transactions with high transaction fee may be still high, and therefore, other miners are motivated to include these transactions, which will increase the transaction inclusion collision and degrade the system throughput. This is the reason why the miners' revenue decreases gradually with the increase of the Bloom filter's false positive probability.
	
	On the other hand, a Bloom filter with a low false positive probability will correspondingly have a large size, leading to a long time to broadcast the signal. Thus, it may have a higher probability of transaction inclusion collision, which will degrade the system throughput. 
	Therefore, it is important to figure out the optimal configuration of the Bloom filter. As shown in Figure \ref{fig:bf_size_fsr}, the Bloom filter with false positive probability as ``0.02'' can achieve the good performance. 
	
	\section{Security Discussions}\label{sec:security}
	
	In this section, we consider several possible security threats caused by  TIPS in DAG-based blockchain and show that TIPS can maintain the system security in long term. Based on the aforementioned design, we know that the miner can not construct a misleading signal easily since each signal contains a PoW in the block header. We are going to analyze two possible security threats including the denial of service attack and the delay of service attack.

	\subsection{Denial of Service  Attack}
	
	There are two types of denial of service attacks in  TIPS. 
	
	\noindent\textbf{BDoS} \cite{mirkin2020bdos}. BDoS is an incentive attack where an adversary can manipulate miners' incentives by broadcasting the block header without publishing the block. However, the ``BDoS'' attack is infeasible in  TIPS due to the following reason. Miners can still obtain a high reward without referring to the latest block. Furthermore, there is a hard-coded timeout for the block header, which can further defend against the BDoS attack.
	
	\noindent\textbf{Signal flood}.
	Another possible attack is that the attacker can broadcast a signal with lots of bits of Bloom filter set to 1 to lower the expected value of transactions, which can reduce the miners' expected reward (even less than the mining cost), and  motivate other miners to stop mining. 
	
	Luckily, we can defend against this attack efficiently by checking the number of bits set to 1 in the Bloom filter and rejecting the Bloom filter with too many bits set to 1 when receiving the block header. 
	To see this, we propose a detection indicator of the Bloom filter, which is analyzed as follows.
	After $n$ transactions have been added to the Bloom filter, let $q$ be the fraction of the $b$ bits that are set to 0, i.e., the number of bits still set to 0 is $qb$. The probability that all $h$ hash functions find that their bits are set to 1 is $(1 - q)^h$. Therefore, the expectation of $q$ is 
	\begin{equation}\nonumber
		\mathbb{E}(q) = \left( 1 - \frac{1}{b} \right)^{hn}.
	\end{equation}
	According to \cite{mitzenmacher2017probability}, we have 
	\begin{equation}\nonumber
		P\left(  \mathbb{E}(q) - q \geq \frac{\xi}{m}  \right) \leq  \exp(- 2\xi^2 / kn).
	\end{equation}
	A Bloom filter with too many bits set to 1 will be rejected. Let $\eta$ be the probability of rejecting a valid Bloom filter. Let $X$ be the number of bits that are set to 1 in the Bloom filter. Then the Bloom filter will be rejected if the following condition holds:
	\begin{equation}\label{eq:bloomfilter_too_many_one}
		X \geq b -  b\left( 1 - \frac{1}{b} \right)^{hn} + \sqrt{-\frac{1}{2} hn \ln \eta}  .
	\end{equation}
	
	As an example, for a Bloom filter with $b=20000$ bits, $k=5$ hash functions and $n=2000$ transactions included, if the probability of rejecting a valid Bloom filter is $0.01\%$, a Bloom filter will be rejected if the number of bits that are set to 1 is greater than 9535, while the expectation of the number of bits that are set to 1 for a valid Bloom filter is 7869, which indicates the high sensitivity and accuracy of detection  indicator (\ref{eq:bloomfilter_too_many_one}). 
	Besides, this extra validation is computationally cheap. Condition (\ref{eq:bloomfilter_too_many_one}) can be pre-computed and store in the memory. Even with simple implementation on a PC, the time to check the validation of Bloom filter is 1 $\mu s$ (i.e., 0.99ms per 1000 Bloom filters on a single core). Therefore, TIPS can efficiently defend against this attack.

	\subsection{Delay of Service Attack }
	
	The traditional delay of service attack in DAG-based blockchain has been discussed in \cite{lewenberg2015inclusive}. Different from the traditional attack, there exists another delay of service attack in TIPS.
	
	In TIPS, the attacker can delay the successful record of a transaction $tx_i$ by continuously mining a valid block that includes this transaction, but only broadcast the signal without the whole block. This attack will motivate other miners to avoid including the transaction $tx_i$ in their blocks because the  signal sent by the attacker will lower the expected value of transaction $tx_i$. However, once the signal is expired, other miners will have the motivation to include the transaction $tx_i$ in the block again. Once other miners include the transaction $tx_i$, this delay of service attack is terminated. 

	To achieve tractable analysis, we consider the following scenario. 
	The expiration time for a block header is $T$. The fraction of the computing power of the attacker is $\alpha$. 
	To delay the transaction as long as possible, the attacker needs to keep mining new blocks containing the same transaction before the signal is expired.
	Denote the expected delay time after the attacker initiates this attack as $\mathbb{E}(D)$. 
	If the attacker mines a new block at time $t < T$ before the previous signal is expired, he can delay the transaction with extra $\mathbb{E}(D)$ time. Otherwise, he can only delay at most the expiration time $T$.
	Then we have
	\begin{equation}\nonumber
		\begin{aligned}
			\mathbb{E}(D) &= \int_{0}^{T} (t + \mathbb{E}(D) ) \alpha \lambda e^{- \alpha \lambda t} dt  + \int_{T}^{\infty} T \alpha \lambda e^{- \alpha \lambda t} dt .
		\end{aligned}
	\end{equation}
	Solving the above equation, we have $\mathbb{E}(D) = \frac{ \left( e^{\alpha \lambda T} - 1 \right) }{ \alpha \lambda  }.$
	Therefore, the expected delay time  is
	\begin{equation}\nonumber
		\text{Delay}(\alpha) = \alpha \mathbb{E}(D) = \frac{1}{\lambda} \left( e^{\alpha \lambda T} - 1 \right),
	\end{equation}
	which implies that the expected delay time for this attack is limited, because generally the expiration time for the block header $T$ is small.
	From the attacker's perspective, the attacker can never profit from this delay of service attack. Besides, the longer the attack goes on, the greater its cost. 
	From the user's perspective, if a user finds his transaction delayed by this attack, he can increase the transaction fee using the replace-by-fee mechanism \cite{elrom2019bitcoin}. A higher transaction fee will motivate other miners to include this transaction and compete against the attacker. Therefore, we can conclude that this attack has an insignificant impact and can be defended against efficiently.

	\section{Related Work}\label{sec:related}
	
	
	\subsection{Transaction Inclusion Protocol}
	
	In the inclusive  protocols \cite{lewenberg2015inclusive}, the authors model the transaction inclusion as a non-cooperative game with imperfect information, and propose a myopic strategy, which can achieve both high throughput and high quality of service levels.
	Conflux \cite{li2018scaling} models the transaction inclusion as a cooperative game and distributes the transaction fee based on Shapley value \cite{roth1988shapley}. 
	The existing transaction inclusion protocols can be considered as the compromise solution facing the dilemmas in DAG-based blockchain. Along a different line, we propose the novel protocol, TIPS, which can make a breakthrough in these dilemmas and achieve near-optimal performance.
	
	\subsection{Header First Propagation}
	
	In the current blockchain systems like Bitcoin and Ethereum, miners generally broadcast the block header before broadcasting the whole block, which can help to avoid the repeated propagation for the same block and thus speed up the propagation process \cite{decker2013information}. However, this header first propagation might lead to the SPV mining, which will threaten the system security \cite{mirkin2020bdos}. 
	In TIPS, we also propagate the block header first but for different purposes. We embedded the ``signal'' in the block header to avoid transaction inclusion collision and boost the system performance. We also show that TIPS can maintain system security in long term.
	
	\subsection{Bloom Filter in Blockchain}
	
	Bloom filter and its variant have been previously adopted in the blockchain system as auxiliaries, especially in log check-up. \cite{wang2020abacus,ma2019blockchain}. In Bitcoin \cite{nakamoto2019bitcoin}, the SPV node can help to limit the amount of transaction data they receive from full nodes to those transactions that affect their wallet while maintaining privacy. In Ethereum \cite{wood2014ethereum}, the Bloom filter of the receipt logs can help nodes to access log data efficiently and securely. Besides, Bloom filter and its variant like invertible bloom lookup table \cite{O1_propagation} can help to compress the block size. 
	To compress the block in blockchain system, the false positive probability of the Bloom filter in \cite{O1_propagation,ozisik2019graphene} should be small enough, which will lead to a large size of Bloom filter. For example, Graphene's block announcements are $1/10$ the size of the whole blocks \cite{ozisik2019graphene}.
	
	Different from previous work, the Bloom filter plays a crucial role in the protocol, serving as a signal in TIPS. Given the compact size of the Bloom filter-based signal, we can broadcast the necessary information to the whole network in a much shorter time. As mentioned in Section \ref{sec:bloom}, the size of the signal in our protocol can be smaller than $1/100$ the size of the whole block. Furthermore, the mitigation of block propagation mentioned in Section \ref{sec:system_mode:block_propagation} also helps to broadcast the signal faster.
	
	\subsection{Block Transmission Optimization}
	
	TIPS improves the system performance through drastically shortening the effective network delay with the signaling approach. There are several works focusing on optimizing the block transmission delay, which also helps to shorten the network delay. In \cite{bi2018accelerated}, the authors propose an accelerated method for block propagation by selecting proper neighbors. \cite{kan2018boost} boosts blockchain broadcast propagation with tree routing. Coded design \cite{zhang2021speeding} and compacting technology \cite{ozisik2019graphene} are also used to speed up the block propagation.
	In fact, TIPS can be used along with these approaches to improve the system performance of the DAG-based blockchain.

	\section{Conclusion}\label{sec:conclusion}
	
	In this paper, we proposed a novel Transaction Inclusion Protocol with Signaling, TIPS, which can explicitly resolve the dilemmas in DAG-based blockchain and achieve near-optimal performance while maintaining system security. Both the theoretical analysis and the experimental results significantly demonstrate the high efficiency of TIPS.
	
	There are several interesting directions to explore in the future, such as how to simplify the ``signal'' and make the ``signal'' broadcast to the whole network faster. Besides, it is interesting to analyze the dynamic game between different miners with repeated interaction in the long run considering the tit-for-tat property. \ADD{Moreover, it is meaningful to provide theoretical analysis for the challenging scenario where the network is asynchronous and the miners are heterogeneous.}

	\bibliographystyle{IEEEtran}
	\bibliography{IEEEabrv,reference}

\begin{thebibliography}{10}
\providecommand{\url}[1]{#1}
\csname url@samestyle\endcsname
\providecommand{\newblock}{\relax}
\providecommand{\bibinfo}[2]{#2}
\providecommand{\BIBentrySTDinterwordspacing}{\spaceskip=0pt\relax}
\providecommand{\BIBentryALTinterwordstretchfactor}{4}
\providecommand{\BIBentryALTinterwordspacing}{\spaceskip=\fontdimen2\font plus
\BIBentryALTinterwordstretchfactor\fontdimen3\font minus
  \fontdimen4\font\relax}
\providecommand{\BIBforeignlanguage}[2]{{%
\expandafter\ifx\csname l@#1\endcsname\relax
\typeout{** WARNING: IEEEtran.bst: No hyphenation pattern has been}%
\typeout{** loaded for the language `#1'. Using the pattern for}%
\typeout{** the default language instead.}%
\else
\language=\csname l@#1\endcsname
\fi
#2}}
\providecommand{\BIBdecl}{\relax}
\BIBdecl

\bibitem{swan2019blockchain}
M.~Swan, J.~Potts, S.~Takagi, F.~Witte, and P.~Tasca, \emph{Blockchain
  economics: implications of distributed ledgers: markets, communications
  networks, and algorithmic reality}.\hskip 1em plus 0.5em minus 0.4em\relax
  World Scientific Publishing Co. Pte. Ltd., 2019.

\bibitem{xiong2018cloud}
Z.~Xiong, S.~Feng, W.~Wang, D.~Niyato, P.~Wang, and Z.~Han, ``Cloud/fog
  computing resource management and pricing for blockchain networks,''
  \emph{IEEE Internet of Things Journal}, vol.~6, no.~3, pp. 4585--4600, 2018.

\bibitem{kang2019toward}
J.~Kang, Z.~Xiong, D.~Niyato, D.~Ye, D.~I. Kim, and J.~Zhao, ``Toward secure
  blockchain-enabled internet of vehicles: Optimizing consensus management
  using reputation and contract theory,'' \emph{IEEE Transactions on Vehicular
  Technology}, vol.~68, no.~3, pp. 2906--2920, 2019.

\bibitem{zhou2020solutions}
Q.~Zhou, H.~Huang, Z.~Zheng, and J.~Bian, ``Solutions to scalability of
  blockchain: A survey,'' \emph{IEEE Access}, vol.~8, pp. 16\,440--16\,455,
  2020.

\bibitem{lewenberg2015inclusive}
Y.~Lewenberg, Y.~Sompolinsky, and A.~Zohar, ``Inclusive block chain
  protocols,'' in \emph{International Conference on Financial Cryptography and
  Data Security}.\hskip 1em plus 0.5em minus 0.4em\relax Springer, 2015, pp.
  528--547.

\bibitem{wang2021understanding}
T.~Wang, Q.~Wang, Z.~Shen, Z.~Jia, and Z.~Shao, ``Understanding characteristics
  and system implications of {DAG}-based blockchain in {IoT} environments,''
  \emph{IEEE Internet of Things Journal}, 2021.

\bibitem{cao2019internet}
B.~Cao, Y.~Li, L.~Zhang, L.~Zhang, S.~Mumtaz, Z.~Zhou, and M.~Peng, ``When
  {internet of things} meets blockchain: Challenges in distributed consensus,''
  \emph{IEEE Network}, vol.~33, no.~6, pp. 133--139, 2019.

\bibitem{zhang2021dag}
H.~Zhang, S.~Leng, F.~Wu, and H.~Chai, ``A {DAG} blockchain enhanced
  user-autonomy spectrum sharing framework for {6G}-enabled {IoT},'' \emph{IEEE
  Internet of Things Journal}, 2021.

\bibitem{yang2020ldv}
W.~Yang, X.~Dai, J.~Xiao, and H.~Jin, ``{LDV}: A lightweight {DAG}-based
  blockchain for vehicular social networks,'' \emph{IEEE Transactions on
  Vehicular Technology}, vol.~69, no.~6, pp. 5749--5759, 2020.

\bibitem{li2020direct}
Y.~Li, B.~Cao, M.~Peng, L.~Zhang, L.~Zhang, D.~Feng, and J.~Yu, ``Direct
  acyclic graph-based ledger for internet of things: Performance and security
  analysis,'' \emph{IEEE/ACM Transactions on Networking}, vol.~28, no.~4, pp.
  1643--1656, 2020.

\bibitem{gupta2019cdag}
H.~Gupta and D.~Janakiram, ``{CDAG}: A serialized blockdag for permissioned
  blockchain,'' \emph{arXiv preprint arXiv:1910.08547}, 2019.

\bibitem{gobel2017increased}
J.~G{\"o}bel and A.~E. Krzesinski, ``Increased block size and {Bitcoin}
  blockchain dynamics,'' in \emph{2017 27th International Telecommunication
  Networks and Applications Conference (ITNAC)}.\hskip 1em plus 0.5em minus
  0.4em\relax IEEE, 2017, pp. 1--6.

\bibitem{bloom1970space}
B.~H. Bloom, ``Space/time trade-offs in hash coding with allowable errors,''
  \emph{Communications of the ACM}, vol.~13, no.~7, pp. 422--426, 1970.

\bibitem{hari2019accel}
A.~Hari, M.~Kodialam, and T.~Lakshman, ``Accel: Accelerating the {Bitcoin}
  blockchain for high-throughput, low-latency applications,'' in \emph{IEEE
  infocom 2019-IEEE conference on computer communications}.\hskip 1em plus
  0.5em minus 0.4em\relax IEEE, 2019, pp. 2368--2376.

\bibitem{chen2020nonlinear}
L.~Chen, L.~Xu, Z.~Gao, A.~Sunny, K.~Kasichainula, and W.~Shi, ``Nonlinear
  blockchain scalability: a game-theoretic perspective,'' \emph{arXiv preprint
  arXiv:2001.08231}, 2020.

\bibitem{eyal2014majority}
I.~Eyal and E.~G. Sirer, ``Majority is not enough: {Bitcoin} mining is
  vulnerable,'' in \emph{International conference on financial cryptography and
  data security}.\hskip 1em plus 0.5em minus 0.4em\relax Springer, 2014, pp.
  436--454.

\bibitem{O1_propagation}
\BIBentryALTinterwordspacing
G.~Andresen. (2014) {O(1) Block Propagation}. [Online]. Available:
  \url{https://gist.github.com/gavinandresen/e20c3b5a1d4b97f79ac2}
\BIBentrySTDinterwordspacing

\bibitem{bi2018accelerated}
W.~Bi, H.~Yang, and M.~Zheng, ``An accelerated method for message propagation
  in blockchain networks,'' \emph{arXiv preprint arXiv:1809.00455}, 2018.

\bibitem{bitcoinvisuals}
\BIBentryALTinterwordspacing
{Fees Per Tx (BTC) in Bitcoin Visuals}. [Online]. Available:
  \url{https://bitcoinvisuals.com/chain-fees-tx-btc}
\BIBentrySTDinterwordspacing

\bibitem{decker2013information}
C.~Decker and R.~Wattenhofer, ``Information propagation in the {Bitcoin}
  network,'' in \emph{IEEE P2P 2013 Proceedings}.\hskip 1em plus 0.5em minus
  0.4em\relax IEEE, 2013, pp. 1--10.

\bibitem{matloff2008introduction}
N.~Matloff, ``Introduction to discrete-event simulation and the simpy
  language,'' \emph{Davis, CA. Dept of Computer Science. University of
  California at Davis. Retrieved on August}, vol.~2, no. 2009, pp. 1--33, 2008.

\bibitem{li2018scaling}
C.~Li, P.~Li, D.~Zhou, Z.~Yang, M.~Wu, G.~Yang, W.~Xu, F.~Long, and A.~C.-C.
  Yao, ``A decentralized blockchain with high throughput and fast
  confirmation,'' in \emph{2020 $\{$USENIX$\}$ Annual Technical Conference
  ($\{$USENIX$\}$$\{$ATC$\}$ 20)}, 2020, pp. 515--528.

\bibitem{onlineBloom}
\BIBentryALTinterwordspacing
pybloom-live 3.1.0. [Online]. Available:
  \url{https://pypi.org/project/pybloom-live/}
\BIBentrySTDinterwordspacing

\bibitem{mirkin2020bdos}
M.~Mirkin, Y.~Ji, J.~Pang, A.~Klages-Mundt, I.~Eyal, and A.~Juels, ``{BDoS}:
  Blockchain denial-of-service,'' in \emph{Proceedings of the 2020 ACM SIGSAC
  conference on Computer and Communications Security}, 2020, pp. 601--619.

\bibitem{mitzenmacher2017probability}
M.~Mitzenmacher and E.~Upfal, \emph{Probability and computing: Randomization
  and probabilistic techniques in algorithms and data analysis}.\hskip 1em plus
  0.5em minus 0.4em\relax Cambridge university press, 2017.

\bibitem{elrom2019bitcoin}
E.~Elrom, ``Bitcoin wallets and transactions,'' in \emph{The Blockchain
  Developer}.\hskip 1em plus 0.5em minus 0.4em\relax Springer, 2019, pp.
  121--171.

\bibitem{roth1988shapley}
A.~E. Roth, \emph{The Shapley value: essays in honor of Lloyd S.
  Shapley}.\hskip 1em plus 0.5em minus 0.4em\relax Cambridge University Press,
  1988.

\bibitem{wang2020abacus}
T.~Wang, W.~Zhu, Q.~Ma, Z.~Shen, and Z.~Shao, ``Abacus: Address-partitioned
  bloom filter on address checking for uniqueness in {IoT} blockchain,'' in
  \emph{Proceedings of the 39th International Conference on Computer-Aided
  Design}, 2020, pp. 1--7.

\bibitem{ma2019blockchain}
X.~Ma, L.~Xu, and L.~Xu, ``Blockchain retrieval model based on elastic bloom
  filter,'' in \emph{International Conference on Web Information Systems and
  Applications}.\hskip 1em plus 0.5em minus 0.4em\relax Springer, 2019, pp.
  527--538.

\bibitem{nakamoto2019bitcoin}
S.~Nakamoto, ``Bitcoin: A peer-to-peer electronic cash system,'' Manubot, Tech.
  Rep., 2019.

\bibitem{wood2014ethereum}
G.~Wood \emph{et~al.}, ``Ethereum: A secure decentralised generalised
  transaction ledger,'' \emph{Ethereum project yellow paper}, vol. 151, no.
  2014, pp. 1--32, 2014.

\bibitem{ozisik2019graphene}
A.~P. Ozisik, G.~Andresen, B.~N. Levine, D.~Tapp, G.~Bissias, and S.~Katkuri,
  ``Graphene: efficient interactive set reconciliation applied to blockchain
  propagation,'' in \emph{Proceedings of the ACM Special Interest Group on Data
  Communication}, 2019, pp. 303--317.

\bibitem{kan2018boost}
J.~Kan, L.~Zou, B.~Liu, and X.~Huang, ``Boost blockchain broadcast propagation
  with tree routing,'' in \emph{International Conference on Smart
  Blockchain}.\hskip 1em plus 0.5em minus 0.4em\relax Springer, 2018, pp.
  77--85.

\bibitem{zhang2021speeding}
L.~Zhang, T.~Wang, and S.~C. Liew, ``Speeding up block propagation in
  blockchain network: Uncoded and coded designs,'' \emph{arXiv preprint
  arXiv:2101.00378}, 2021.

\end{thebibliography}
	
	
	\ifx\hiddenAppendix\undefine
	
	\appendices
	
%
	
	\section{Proof of Lemma \ref{lemma:reward}}\label{proof:lemma:reward}
	
	\begin{proof}
		Since the block generation process follows the Poisson process with rate $\lambda$, the probability of generating  $k$ more blocks during the block propagation duration is $\frac{(\lambda \Delta)^k}{k!} e^{-\lambda \Delta}$. Because  other miners include transaction $i$ in their block with a marginal probability $p_i$, the probability that there are $\iota$ more miners including the transaction $i$ is $\binom{k}{\iota} p_i^\iota (1 - p_i)^{k-\iota}$. Therefore, we have
		\begin{footnotesize}
			\begin{equation}\nonumber
				\begin{aligned}
					&r(p_i | \Delta) =   \sum_{k=0}^{\infty} \left( \frac{(\lambda \Delta)^k}{k!} e^{-\lambda \Delta} \right) \sum_{\iota = 0}^{k} \binom{k}{\iota} p_i^\iota (1 - p_i)^{k-\iota} \frac{1}{\iota + 1} \\ 
					&=  \sum_{k=0}^{\infty} \left( \frac{(\lambda \Delta)^k}{k!} e^{-\lambda \Delta} \right) \frac{f_i}{p_i(k +1)} \!\! \sum_{\iota = 0}^{k} \!\! \frac{(k+1)!}{(\iota+1)! (k-\iota)!} p_i^{\iota+1} (1 - p_i)^{k-\iota}  \\
					&=   \sum_{k=0}^{\infty} \left( \frac{(\lambda \Delta)^k}{k!} e^{-\lambda \Delta} \right) \frac{f_i}{p_i(k +1)} \left( 1 - (1 - p_i)^{k+1} \right) \\
					&=  \frac{1}{p_i} \sum_{k=0}^{\infty} \left( \frac{(\lambda\Delta)^k}{(k+1)!} e^{-\lambda \Delta} - \frac{(\lambda \Delta)^k (1-p_i)^{k+1} }{(k+1)!} e^{-\lambda \Delta} \right) \\
					&=  \frac{1}{p_i} \left( \frac{1 - e^{-\lambda \Delta}}{\lambda \Delta} -  \frac{e^{-\lambda \Delta p_i} - e^{-\lambda \Delta}}{\lambda \Delta} \right) = \frac{\left(1 - e^{-\lambda \Delta p_i}\right) }{\lambda \Delta p_i}.
				\end{aligned}
			\end{equation}
		\end{footnotesize}
		Then given that other miners include transaction $i$ in their blocks with the probability $p_i'$, the miner's revenue on one block with strategy $\textbf{p}$ is
		$
		R(\textbf{p} | \textbf{p}^*) = \sum_{i=1}^{m} p_i f_i r(p_i^*).
		$
		Therefore, when all the miners adopt the symmetric transaction inclusion strategy $\textbf{p}$, we know that the miners' revenue is 
		$
		R(\textbf{p}) = R(\textbf{p} | \textbf{p}) = \sum_{i=1}^{m} p_i f_i r(p_i|\Delta).
		$
		The proof is thus completed.
	\end{proof}

	\section{Proof of Theorem \ref{th:approx_rand}}\label{proof:th:approx_rand}
	
	\begin{proof}
		Based on Lemma \ref{lemma:reward}, the expectation of total transaction fee reward of a miner who adopts inclusion strategy $\textbf{p}$ given that other miners adopt the strategy $\textbf{p}^{\rm rand}$ is
		\begin{equation}\nonumber
			R(\textbf{p}|\textbf{p}^{\text{rand}}) = \sum_{i=1}^{m} p_i f_i r, \text{ where } r = \frac{1 - e^{-\lambda \Delta \frac{n}{m}}}{\lambda \Delta \frac{n}{m}}.
		\end{equation} 
		Then we have
		\begin{small}
			\begin{equation}\nonumber
				\begin{aligned}
					&\max_{\textbf{p} \in \mathbb{P}} R(\textbf{p}|\textbf{p}^{\text{rand}}) - R(\textbf{p}^{\text{rand}}) = R(\textbf{p}_{\text{top}}|\textbf{p}^{\text{rand}}) - R(\textbf{p}^{\text{rand}}) \\
					=& r \sum_{i=1}^{n} f_i - \frac{n}{m} r \sum_{i=1}^{m} f_i = \frac{1 - e^{-\lambda \Delta \frac{n}{m}}}{\lambda \Delta \frac{n}{m}} \left( \frac{1}{n} \sum_{i=1}^{n} f_i - \frac{1}{m} \sum_{i=1}^{m} f_i   \right).
				\end{aligned}
			\end{equation}
		\end{small}
		Besides, we denote $h(\Delta) = \frac{1 - e^{-\lambda \Delta \frac{n}{m}}}{\lambda \Delta \frac{n}{m}}$. Note that $h(\Delta)$ is monotonically decreasing in $\Delta$. Specially, we have that $\lim_{\Delta \rightarrow \infty} h(\Delta) = 0$. Therefore, when $\Delta \rightarrow \infty$, the random strategy is the Nash equilibrium.
		The proof is thus completed.
	\end{proof}
	
	\section{Proof of Theorem \ref{th_tps}}\label{proof:pro_tps}
	
	\begin{proof}
		During the  propagation time $\Delta$, the miner who mines the new block will broadcast it to the network, with other miners keep mining without any notification of the latest block. The probability of other miners to mine $k$ more blocks during the block propagation time is 
		\begin{equation}\label{eq:Poisson}
			P(k, \Delta) = \frac{(\lambda \Delta)^k}{k!} e^{-\lambda \Delta}.
		\end{equation}
		For convenience, we define that
		\begin{equation}\nonumber
			\delta_i = \begin{cases}
				1, & \!\! \text{the } i \text{-th transaction is included in one of the blocks} ,\\
				0, & \!\! \text{otherwise}.
			\end{cases}
		\end{equation}
		Then the probability that the $i$-th transaction in the transaction pool is included by $k+1$ miners in their blocks is
		\begin{equation}\nonumber
			P(\delta_i = 1, k+1) = 1 - \left( 1 - p_i \right)^{k+1}.
		\end{equation}
		Therefore, the expected number of transactions included during the  propagation time with additional $k$ blocks emerging is 
		\begin{equation}\nonumber
			\begin{aligned}
				\mathbb{E}(X, k) &= \mathbb{E} \left(\sum_{i=1}^{m} P (\delta_i, k) \right) = \sum_{i=1}^{m} \left( 1 - \left( 1 - p_i \right)^{k+1} \right).
			\end{aligned}
		\end{equation}
		Therefore, the average block capacity utilization in long term is 
		\begin{equation}\nonumber
			\begin{aligned}
				U(\textbf{p}) &= \frac{\sum_{k=0}^{\infty} \left( \frac{(\lambda \Delta)^k}{k!} e^{-\lambda \Delta} \right) \mathbb{E}(X, k)}{\sum_{k=0}^{\infty} \left( \frac{(\lambda \Delta)^k}{k!} e^{-\lambda \Delta} \right)(k+1)n} \\
				&= \frac{\sum_{k=0}^{\infty} \left( \frac{(\lambda \Delta)^k}{k!} e^{-\lambda \Delta} \right) \sum_{i=1}^{m} \left( 1 - \left( 1 - p_i \right)^{k+1} \right)}{\sum_{k=0}^{\infty} \left( \frac{(\lambda \Delta)^k}{k!} e^{-\lambda \Delta} \right)(k+1)n} \\
				&= \frac{ e^{-\lambda \Delta}  \sum_{i=1}^{m}\sum_{k=0}^{\infty} \left( \frac{ (\lambda \Delta)^k }{k!} - \frac{(1 - p_i) (\lambda \Delta (1-p_i))^{k} }{k!} \right)  }{n(\lambda \Delta + 1)} \\
				&= \frac{ m - \sum_{i=1}^{m} (1 - p_i) e^{- \lambda \Delta p_i} }{n(\lambda \Delta + 1)}.
			\end{aligned}
		\end{equation}
		Thus, the throughput of the DAG-based blockchain is 
		\begin{equation}\nonumber
			\text{TPS}(\textbf{p}) = \lambda n U(\textbf{p}) = \frac{ \lambda \left( m - \sum_{i=1}^{m} (1 - p_i) e^{-\lambda \Delta p_i} \right) }{(\lambda \Delta + 1)}.
		\end{equation}
		The proof is thus completed.
	\end{proof}
	
	\section{Proof of Theorem \ref{th:approximate}}\label{proof:th:approximate}
	
	\begin{proof}
		Based on Lemma 1, the expectation of total transaction fee reward of a miner with  transaction inclusion strategy $\textbf{p}$, given other miners adopt the strategy $\textbf{p}^{\rm top}$, is 
		\begin{equation}\nonumber
			R(\textbf{p}|\textbf{p}^{\rm top}) =  \frac{1 - e^{-\lambda \tau}}{\lambda \tau} \sum_{i=1}^{n} \ p_i f_i + \sum_{i=n+1}^{m} p_i f_i.
		\end{equation}
		If the miner adopts the top $n$ strategy, his expected reward is 
		\begin{equation}\nonumber
			R(\textbf{p}^{\rm top}) = R(\textbf{p}^{\rm top}|\textbf{p}^{\rm top}) =  \frac{1 - e^{-\lambda \tau}}{ \lambda \tau }\sum_{i=1}^{n} f_i.
		\end{equation}
		Therefore, the extra reward that miner can obtain from deviation is
		\begin{small}
			\begin{equation}\nonumber
				\begin{aligned}
					\max_{\textbf{p} \in \mathbb{P}} R(\textbf{p}|\textbf{p}^{\rm top}) - R(\textbf{p}^{\rm top}) &\leq \max \left\{ nf_{n+1} - n \frac{1-e^{-\lambda \tau}}{\lambda \tau} f_n, 0 \right\} \\
					&\leq \max \left\{ n \left(1 - \frac{1 - e^{-\lambda \tau}}{\lambda \tau}  \right) f_n, 0 \right\}.
				\end{aligned}
			\end{equation}
		\end{small}
		We denote that $g(\tau) = 1 - \frac{1 - e^{-\lambda \tau}}{\lambda \tau}$. Note that $g(\tau)$ is monotonically increasing in $\tau$. Thus we have $\lim_{\tau \rightarrow 0} g(\tau) = 0$. Therefore, when $\tau \rightarrow 0$, the top $n$ strategy is the Nash equilibrium.
		The proof is thus completed.
	\end{proof}
	
	\section{Proof of Theorem \ref{th:NE_top_n}}\label{proof:th:NE_top_n}
	
	\begin{proof}
		To prove that the top $n$ strategy is the equilibrium strategy, we only need to prove that given that other miners adopt the top $n$ strategy, the best response of the miner is exactly the top $n$ strategy. Besides, when fixing other miners' strategies, there exists a pure strategy that can be the best response. 
		Without loss of generality, we consider the scenario where one miner deviates from the top $n$ strategy, and selects the transaction set $\mathbb{B}$, while other miners adopt the top $n$ strategy and select the transaction set $\mathbb{A} = \{1,2,\ldots,n\}$. Then we have $|\mathbb{A}| = |\mathbb{B}| = n$.
		Besides, we denote that $n_1 = | \mathbb{A} \cap \mathbb{B} |$, and $n_2 = \left| \mathbb{A} / \left( \mathbb{A} \cap \mathbb{B} \right) \right|$. Thus, we have $n_1 + n_2 = n$. Then the expected total transaction fee reward for the miner who deviates from the top $n$ strategy and includes the transaction set $\mathbb{B}$ is 
		\begin{small}
			\begin{equation}\nonumber
				\begin{aligned}
					\text{Reward}(\mathbb{B}) &= \text{Tx}(\mathbb{B} / ( \mathbb{B} \cap \mathbb{A}) ) + \sum_{k=0}^{\infty} \left( \frac{(\lambda \tau)^k}{k!} e^{-\lambda \tau} \right) \frac{\text{Tx}(\mathbb{B} \cap \mathbb{A} )}{ k+1 } \\
					&= \text{Tx}(\mathbb{B} / (\mathbb{B} \cap \mathbb{A}) ) + \text{Tx}(\mathbb{B} \cap \mathbb{A} ) \frac{ 1 - e^{-\lambda \tau} }{ \lambda \tau } \\
					&\leq \sum_{i=1}^{n_2} f_{n + i} + \frac{ 1 - e^{-\lambda \tau} }{ \lambda \tau } \sum_{i=1}^{n_1} f_{i},
				\end{aligned}
			\end{equation}
		\end{small}
		and the equality holds when $\mathbb{S} \cap \mathbb{P} = \{ 1, 2, \ldots, n_1 \}$ and $\mathbb{S} / (\mathbb{S} \cap \mathbb{P}) = \{n+1, n+2, \ldots, n+n_2\}$.
		
		Since $\varphi(x)$ is monotonically decreasing, when condition (\ref{eq:NE_condiction}) holds, we have
		\begin{equation}\nonumber
			\frac{1 - e^{-\lambda \tau}}{\lambda \tau} \geq \frac{f_{n+1}}{f_n}.
		\end{equation}
		Besides, since $f_{i+1} \leq f_{i}, \forall i \in \{1, 2, \ldots, m-1\}$, we have
		\begin{equation}\nonumber
			\max_{n_1 = 0, 1, 2, \ldots, n} \frac{ \sum_{i=1}^{n - n_1} f_{n+i} }{ \sum_{i=1}^{n - n_1} f_{n_1 + i} } = \frac{f_{n+1}}{f_n}.
		\end{equation}
		Therefore, when condition (\ref{eq:NE_condiction}) holds, we have
		\begin{equation}\nonumber
			\sum_{i=1}^{n_2} f_{n + i} \leq \frac{ 1 - e^{-\lambda \tau} }{ \lambda \tau } \sum_{i=1}^{n_2} f_{n_1 + i},
		\end{equation}
		which implies that 
		\begin{equation}\nonumber
			\sum_{i=1}^{n_2} f_{n + i} + \frac{ 1 - e^{-\lambda \tau} }{ \lambda \tau } \sum_{i=1}^{n_1} f_{i} \leq \frac{ 1 - e^{-\lambda \tau} }{ \lambda \tau } \sum_{i=1}^{n} f_{i}.
		\end{equation}
		Therefore, we have $\text{Reward}(\mathbb{B}) \leq \text{Reward}(\mathbb{A})$, which implies that $R(\textbf{p}^{\rm top}) \geq \max_{\textbf{p} \in \mathbb{P}} R(\textbf{p} | \textbf{p}^{\rm top})$.
		Thus, the top $n$ strategy is a Nash equilibrium under the condition (\ref{eq:NE_condiction}). Furthermore, we can find that the top $n$ strategy strictly dominates other strategies. Therefore, it is the unique Nash equilibrium. 
	\end{proof}

	\section{Proof of Theorem \ref{th:efficiency:tps}}\label{proof:th:efficiency:tps}

	\begin{proof}
		For convenience, we denote $\textbf{p}^*$ as the equilibrium strategy in the transaction inclusion game. Since the utilization of blocks could not greater than 1, i.e., $U_{\text{optimal}} \leq 1$, we have 
		\begin{equation}\nonumber
			\text{TPS}_{\text{optimal}} = \lambda n U_{\text{optimal}} \leq \lambda n .
		\end{equation}
		Besides, the TPS of the equilibrium strategy is greater or equal to the TPS of the top $n$ strategy, $\textbf{p}_{\rm top}$, then we have
		\begin{equation}\nonumber
			\text{TPS}({\textbf{p}^*}) \geq \text{TPS}(\textbf{p}^{\text{top}}) = \frac{\lambda n}{\lambda \tau + 1} .
		\end{equation}
		Therefore, we have
		\begin{equation}\nonumber
			\text{Efficiency}(\text{TPS}) = \frac{\text{TPS}({\textbf{p}^*})}{\text{TPS}_{\text{optimal}}} \geq \frac{\frac{\lambda n}{\lambda \tau + 1}}{\lambda n} = \frac{1}{\lambda \tau + 1}.
		\end{equation}
		The proof is thus completed.
	\end{proof}	
	
	\section{Proof of Lemma \ref{lemma:limit_tps}}\label{proof:lemma:limit_tps}
	
	\begin{proof}
		For the top $n$ transaction inclusion strategy, we have that $p_1 = p_2 = \cdots = p_n = 1$ and $p_{n+1} = p_{n+2} = \cdots = p_{m} = 0$. According to Theorem \ref{th_tps}, we have
		\begin{equation}\nonumber
			\begin{aligned}
				&\lim_{n \rightarrow \infty} TPS(\textbf{p}^{\text{top}}, n) \\
				=& \lim_{n \rightarrow \infty}  \frac{ \lambda \left( m - \sum_{i=1}^{m} \lambda(1 - p_i) e^{-\lambda \Delta(n) p_i}  \right) }{(\lambda \Delta(n) + 1)} \\
				=& \lim_{n \rightarrow \infty}   { \lambda \left( m - \sum_{i=1}^{m} \lambda(1 - p_i) e^{-\lambda \Delta(n) p_i}  \right) }/{(\lambda \Delta(n) + 1)} \\
				=& \lim_{n \rightarrow \infty}    \frac{\lambda}{{(\lambda \Delta(n) + 1)}} \left( m - \sum_{i=1}^{n} \lambda(1 - 1) e^{-\lambda \Delta(n) 1} \right. \\
				&\left. \quad -  \sum_{i=n+1}^{m} \lambda(1 - 0) e^{-\lambda \Delta(n) 0} \right)  \\
				=&\lim_{n \rightarrow \infty}  \frac{\lambda n}{\lambda \Delta(n) + 1} \\
				=& \frac{1}{\frac{\mathrm{d}\Delta(n)}{\mathrm{d}n}} \quad\quad \text{(Using the L'Hospital's rule)}
			\end{aligned}
		\end{equation}
		The proof is thus completed.
	\end{proof}

	
	\section{Proof of Theorem \ref{th:efficiency:fsr}}\label{proof:th:efficiency:fsr}
	
	Before proving Theorem \ref{th:efficiency:fsr}, we first introduce a new metric ``Fee Service Rate'' (FSR), and then show the relation between the miners' revenue and FSR, finally we can prove Theorem \ref{th:efficiency:fsr} through the analysis of FSR.
	
	Fee service rate (FSR) is defined as the total transaction fee that the blockchain system processes per second. A high FSR implies a high system profit and a good market efficiency in the blockchain system. 
	
	\begin{prop}\label{prop:fsr}
		The fee service rate of the DAG-based blockchain given the transaction inclusion strategy being $\textbf{p}$ and the network propagation delay of a block being $\Delta$ is 
		\begin{equation}
			\text{FSR}(\textbf{p}) = \frac{1}{\Delta} \sum_{i=1}^{m} f_i \left(1 - e^{-\lambda \Delta p_i}  \right),
		\end{equation}
		where $\lambda$ is the block generation rate in the blockchain system.
	\end{prop}
	
	\begin{proof}
		The expectation of the total transaction fees of the transactions included in blocks on the condition that there are $k$ blocks emerging during the network propagation duration is 
		\begin{equation}\nonumber
			\begin{aligned}
				\mathbb{E}(Y, k) &= \mathbb{E} \left(\sum_{i=1}^{m} P (\delta_i, k) f_i \right) = \sum_{i=1}^{m} \left( 1 - \left( 1 - p_i \right)^{k} \right) f_i.
			\end{aligned}
		\end{equation}
		Since the block generation process follows the Poisson process with rate $\lambda$, the probability that there are $k$ blocks emerging during the network propagation duration is $\frac{(\lambda \Delta)^k}{k!} e^{-\lambda \Delta}$.
		Therefore, the fee service rate (FSR) is 
		\begin{equation}\nonumber
			\begin{aligned}
				\text{FSR}(\textbf{p}) &= \frac{1}{\Delta} \sum_{k=0}^{\infty} \left( \frac{(\lambda \Delta)^k}{k!} e^{-\lambda \Delta} \right) \mathbb{E}(Y, k) \\
				& =  \frac{e^{-\lambda \Delta}}{\Delta}  \sum_{i=1}^{m}\sum_{k=0}^{\infty} \left( \frac{ (\lambda \Delta)^k }{k!} - \frac{ (\lambda \Delta (1-p_i))^{k} }{k!} \right)f_i  \\
				&= \frac{e^{-\lambda \Delta}}{\Delta}  \sum_{i=1}^{m} f_i \left( e^{\lambda \Delta} - e^{\lambda \Delta (1 - p_i)} \right) \\
				&= \frac{1}{\Delta} \sum_{i=1}^{m} f_i \left(1 - e^{-\lambda \Delta p_i}  \right).
			\end{aligned}
		\end{equation}
		The proof is thus completed.
	\end{proof}
	
	Combining Proposition \ref{prop:fsr} and Lemma \ref{lemma:reward}, we can find that the miners' revenue is proportional to the FSR, i.e., $R(\textbf{p}) = \text{FSR}(\textbf{p}) / \lambda$, where $\lambda$ is the block generation rate. 
	Therefore, to prove Theorem \ref{th:efficiency:fsr}, we only need to prove that 
	\begin{equation}\nonumber
		\begin{aligned}
			&  \text{Efficiency}(FSR) \geq  \frac{(1 - e^{-\lambda \tau})}{\left( \lfloor \frac{m}{n} \rfloor - \sum_{k=0}^{\lfloor \frac{m}{n} \rfloor} \frac{(\lambda \tau)^k}{k!} e^{-\lambda \tau} \right)} .
		\end{aligned}
	\end{equation}
	where
	$
	\rm Efficiency(FSR) = \frac{\rm FSR \ of \ Equilibrium}{\rm FSR \  of \ optimal \ strategy}.
	$
	The proof goes as follows.
	
	\begin{proof}
		For convenience, we denote $\textbf{p}^*$ as the equilibrium strategy in the transaction inclusion game. We consider the total transaction fees of transactions included in blocks during the network propagation delay of signal $\tau$. The highest possible FSR can be achieved when there is no transaction inclusion collision and the included transactions have the highest transaction fees. Therefore, we have
		\begin{small}
			\begin{equation}\nonumber
				\begin{aligned}
					&\text{FSR}_{\text{optimal}} \leq \frac{1}{\tau} \sum_{k=0}^{\infty} \left( \frac{(\lambda \tau)^k}{k!} e^{-\lambda \tau} \sum_{i=1}^{\max\{kn, m\} } f_i \right) \\ 
					&= \frac{1}{\tau} \left( \left(\sum_{i=0}^{n} f_i\right) \sum_{k=0}^{\infty}\frac{(\lambda \tau)^k}{k!} e^{-\lambda \tau} + \left(\sum_{i=n+1}^{2n} f_i\right) \sum_{k=1}^{\infty}\frac{(\lambda \tau)^k}{k!} e^{-\lambda \tau}  \right. \\
					&\left. + \cdots + \left(\sum_{i=n\lfloor \frac{m}{n} \rfloor  }^{m} f_i\right) \sum_{k=\lfloor \frac{m}{n} \rfloor}^{\infty}\frac{(\lambda \tau)^k}{k!} e^{-\lambda \tau} \right) \\
					&\leq \frac{1}{\tau} \sum_{i=0}^{n} f_i \left( \lfloor \frac{m}{n} \rfloor - \sum_{k=0}^{\lfloor \frac{m}{n} \rfloor} \frac{(\lambda \tau)^k}{k!} e^{-\lambda\tau} \right).
				\end{aligned}
			\end{equation}
		\end{small}
		Besides, since $\textbf{p}^*$ is the equilibrium strategy, we have
		\begin{footnotesize}
			\begin{equation}\nonumber
				R(\textbf{p}^* | \textbf{p}^*) = \sum_{i=1}^{m} p_i^* f_i r(p_i^*) \geq R(\textbf{p}' | \textbf{p}^*) = \sum_{i=1}^{m} p'_i f_i r(p_i^*), \forall \textbf{p}' \in \mathbb{P}.
			\end{equation}
		\end{footnotesize}
		Without loss of generality, we have
		\begin{equation}\nonumber
			R(\textbf{p}^* | \textbf{p}^*) = \sum_{i=1}^{m} p_i^* f_i r(p_i^*) \geq R(\textbf{p}^{\text{top}} | \textbf{p}^*) = \sum_{i=1}^{n}  f_i r(p_i^*).
		\end{equation}
		Since $r(p_i^*)$ is monotonically decreasing, we have that $r(p_i) \geq r(1)$. Therefore, we have
		\begin{equation}\nonumber
			R(\textbf{p}^{\text{top}} | \textbf{p}^*) = \sum_{i=1}^{n}  f_i r(p_i^*) \geq R(\textbf{p}^{\text{top}} | \textbf{p}_{\text{top}}) = \sum_{i=1}^{n}  f_i r(1).
		\end{equation}
		Thus we have that 
		\begin{equation}\nonumber
			R(\textbf{p}^*) \geq R(\textbf{p}^{\text{top}}) =  \sum_{i=1}^{n}  f_i r(1) = \frac{(1 - e^{-\lambda\tau})}{\lambda\tau} \sum_{i=1}^{n}  f_i .
		\end{equation}
		Besides, we have proved that 
		\begin{equation}\nonumber
			\text{FSR}(\textbf{p}') = \lambda R(\textbf{p}'), \quad \forall \textbf{p}' \in \mathbb{P}.
		\end{equation}
		Thus, we have
		\begin{equation}\nonumber
			\text{FSR}(\textbf{p}^*) \geq \lambda R(\textbf{p}^{\text{top}}) = \frac{(1 - e^{-\lambda\tau})}{\tau} \sum_{i=1}^{n} f_i .
		\end{equation}
		Therefore, we have   
		\begin{equation}\nonumber
			\text{Efficiency}(\text{FSR}) = \frac{\text{FSR}(\textbf{p}^*)}{\text{FSR}_{\text{optimal}}} \geq \frac{(1 - e^{-\lambda\tau})}{\left( \lfloor \frac{m}{n} \rfloor - \sum_{k=0}^{\lfloor \frac{m}{n} \rfloor} \frac{(\lambda \tau)^k}{k!} e^{-\lambda\tau} \right)} .
		\end{equation}
		The proof is thus completed.
	\end{proof}

	\else
	\fi

	\vfill
	
\end{document}